\documentclass[a4paper,UKenglish,cleveref, autoref, thm-restate]{lipics-v2021}

\title{Decomposition of Automata Recognizing Ideals}

\author{Mathias Berry}{Université Marie et Louis Pasteur, CNRS, institut FEMTO-ST, F- 25000 Besançon, France}{mathias.berry@femto-st.fr}{}{}
\author{Pierre-Cyrille H{\'e}am}{Université Marie et Louis Pasteur, CNRS, institut FEMTO-ST, F- 25000 Besançon, France }{pierre-cyrille.heam@femto-st.fr}{}{}
\author{Isma{\"{e}}l Jecker}{Université Marie et Louis Pasteur, CNRS, institut FEMTO-ST, F- 25000 Besançon, France }{ismael.jecker@femto-st.fr}{}{This research was partially funded by the Agence Nationale de la Recherche (ANR) grant ANR-25-CE48-2447 FAVOR}

\authorrunning{M.Berry, P.-C.H{\'e}am, I.Jecker}

\Copyright{Mathias Berry, Pierre-Cyrille H{\'e}am, Isma{\"{e}}l Jecker}



\nolinenumbers
\hideLIPIcs

\EventEditors{John Q. Open and Joan R. Access}
\EventNoEds{2}
\EventLongTitle{42nd Conference on Very Important Topics (CVIT 2016)}
\EventShortTitle{CVIT 2016}
\EventAcronym{CVIT}
\EventYear{2016}
\EventDate{December 24--27, 2016}
\EventLocation{Little Whinging, United Kingdom}
\EventLogo{}
\SeriesVolume{42}
\ArticleNo{23}

\usepackage{amsthm} 
\usepackage{graphicx} 
\usepackage[dvipsnames]{xcolor}
\usepackage{stmaryrd} 
\usepackage{extarrows}  
\usepackage[capitalise]{cleveref} 
\usepackage[utf8]{inputenc}
\usepackage[T1]{fontenc} 
\usepackage{amsmath,amsfonts,amssymb} 
\usepackage{dsfont} 
\usepackage{thmtools}  
\usepackage{thm-restate} 
\usepackage{todonotes}	
\usepackage{algorithm2e} 
\usepackage{pdfpages}
\usepackage[appendix=append]{apxproof} 
\usepackage{wrapfig} 

\usetikzlibrary{arrows}
\usetikzlibrary{automata}
\usetikzlibrary{shapes}
\usetikzlibrary[automata] 
\usetikzlibrary {arrows.meta,automata,positioning}
\usetikzlibrary{decorations.pathmorphing}
\usetikzlibrary{calc}

\newcommand{\N}[0]{\mathbb{N}}

\newcommand{\setint}[1]{\llbracket 1, #1 \rrbracket}
\newcommand{\setintz}[1]{\llbracket 0, #1 \rrbracket}
\newcommand{\set}[2]{\left\{#1 \mid #2 \right\}}

\newcommand{\A}[0]{\mathcal{A}}

\newcommand{\Auto}[0]{(S,\Sigma, \iota, F, \delta)}
\newcommand{\AutoLin}[0]{(\{q_0, q_1, \ldots, q_n\},\Sigma, q_0, \{q_n\}, \delta)}
\newcommand{\AutoLinGen}[0]{(\{q_0, q_1, \ldots, q_n\},\Sigma, q_0, F, \delta)}
\newcommand{\minAuto}[0]{(S,\Sigma, \iota, \{q_f\}, \delta)}
\newcommand{\I}[0]{\mathcal{I}}
\newcommand{\Res}[2]{R(#1,#2)}

\newcommand{\norm}[1]{{\rm rk}(#1)}
\newcommand{\sep}[0]{q_{\rm sep}}
\newcommand{\setsep}[0]{S_{\rm sep}}

\newcommand{\wit}[0]{w_{\rm wit}}
\newcommand{\Wit}[0]{S_{\rm wit}}
\newcommand{\pred}[1]{{\rm Anc}(#1)}
\newcommand{\predA}[2]{{\rm Anc}_{#1}(#2)}
\newcommand{\succs}[1]{{\rm Desc}(#1)}
\newcommand{\succsA}[2]{{\rm Desc}_{#1}(#2)}
\newcommand{\mifa}[1]{{\rm Fam}(#1)}
\newcommand{\mifaA}[2]{{\rm Fam}_{#1}(#2)}

\newcommand{\Lmin}[0]{L_{\min}}
\newcommand{\wmax}[0]{w}

\newdimen\demi%
\def\shu#1#2#3{%
\dimen\demi=#1pt%
\divide\dimen\demi by 2%
\def\haut{\vrule width#3pt height#2pt depth0pt}%
\def\bas{\vrule width\dimen\demi height#3pt depth0pt}%
\hbox{\haut\bas\haut\bas\haut}}
\def\shuffle{\mathrel{\mathchoice{\shu{8}{4}{.4}}{\shu{8}{4}{.4}}{\shu{6}{3}{.4}}{\shu{3}{1.5}{.4}}}}

\newcommand{\Shuffle}[1]{#1 \shuffle \ \Sigma^*}
\newcommand{\lang}[1]{\mathcal{L}(#1)}

\SetKwComment{Comment}{/* }{ */} 
\RestyleAlgo{ruled} 

\newcommand{\partAlph}[2]{\Sigma_{#1,#2}}
\newtheoremrep{prop}[proposition]{Proposition}
\newtheoremrep{lem}[lemma]{Lemma}
\newtheoremrep{theo}[theorem]{Theorem}
\newtheoremrep{coro}[corollary]{Corollary}

\renewcommand{\leq}{\leqslant}
\renewcommand{\geq}{\geqslant}

\begin{document}

\maketitle

\begin{abstract}
	Minimizing the size of finite automata is a fundamental problem in theoretical computer science.
	Beyond standard minimization, further reductions can be achieved
	by decomposing an automaton into smaller components whose languages
	combine via intersection or union to recover the original language.
	However, in general, no polynomial-time algorithm is known for computing such decompositions.
	
  	In this paper, we focus on
  	automata that recognize ideals, that is, 
  	languages  at level 1/2  in the Straubing–Thérien
  	hierarchy.
  	Equivalently, these languages are expressible as a finite union of languages of
  	the form \(\Sigma^*a_1\Sigma^*\dots\Sigma^*a_n\Sigma^*\) where
  	\(\Sigma\) is an alphabet and \(a_i\) are letters of
  	\(\Sigma\).
  We show that the two problems of deciding whether an automata recognizing an ideal  
  can be decomposed
  into an intersection or a union of smaller automata are decidable in NL.
	Moreover, we provide a polynomial-time algorithm
	that computes a decomposition into an intersection, if one exists,
	while ensuring that the resulting components also recognize ideal languages.
\end{abstract}

\section{Introduction}

Finite automata are fundamental tools in computer science,
with applications ranging from text processing to system modelling and
machine learning, thanks to their strong algorithmic
properties. In contexts such as model verification, automata
can grow prohibitively large, limiting applicability to industrial-scale
problems. Significant research has focused on
automata minimization (see for instance
\cite{DBLP:books/aw/HopcroftU79,DBLP:journals/siamcomp/JiangR93,DBLP:journals/tc/KamedaW70,brainerd1968minimalization}),
yet an alternative approach consists in decomposing complex systems into simpler components~\cite{DBLP:conf/compos/1997}.

In this context, the notion of deterministic automaton
decomposition was introduced by O. Kupferman and
J. Mosheiff~\cite{kupferman2015prime}.
In this paper, we push this line of research further and distinguish two types of decompositions.
An intersection [resp. union] decomposition of an automaton \(\A\) is a
finite collection of smaller automata \(\A_1,\ldots,\A_n\) such that the
intersection [resp. union] of their recognized languages coincides
with that \(\A\). When no such decomposition exists, \(\A\) is called prime. Such decompositions reduce the
complexity of models or specifications by splitting them into simpler
components.
For instance,
verifying that a system satisfies a specification that has an intersection decomposition
reduces to verifying whether the system satisfies each component individually.
This approach can improve computational efficiency, for instance
through distributed computation
\cite{DBLP:conf/cav/SternD97,DBLP:books/sp/18/BarnatBDLPPR18,DBLP:conf/formats/DalsgaardLLOP12}.

We focus on decomposition problems for the class of \emph{ideal languages},
that occupy level 1/2 in the Straubing–Thérien
hierarchy~\cite{Haines,straubing1985finite,therien1981classification}.
Formally, this the class of languages expressible as a
finite union of languages of the form \(\Sigma^*a_1\Sigma^*\dots
\Sigma^* a_n \Sigma^*\) where \(\Sigma\) is a finite alphabet and
\(a_i\) are letters. 
Ideal languages are also known as \emph{shuffle ideals}~\cite{heam2002shuffle}, or simply \emph{ideals},
and they arise naturally in the study of word combinatorics~\cite{higman}
and the formal verification of systems with communication channels~\cite{abdulla2004using}.
They also correspond to the languages definable by the existential
fragment of FO(<)~\cite{barloy2022regular},
and  generate, via Boolean combinations, the class of
\emph{piecewise testable} languages~\cite{DBLP:conf/automata/Simon75,DBLP:journals/dm/Klima11,DBLP:conf/csl/KarandikarS16,DBLP:journals/lmcs/MasopustK21}.
Moreover, ideal languages are upward closed under the subword relation,
a property that has been studied in relation to state complexity~\cite{karandikar2014state,okhotin2010state}
and combinatorial properties of words~\cite{karandikar2015index}.
This property underlies the relevance of the class of ideals.

\subparagraph*{Contributions.}
We study both intersection and union decompositions
for automata recognizing ideals, focusing on the complexity of deciding
primality and on the effective construction of decompositions. Note that these two problems are interreducible in general, since an automaton is prime for intersection if and only if its complement is prime for union. However, ideal languages are not closed by complement, thus these two problems are different in that case.
Our results first address intersection decomposition.

\begin{restatable}[]{theorem}{maintheorem}\label{theo:space complexity}
	Given a minimal automaton \(\A\) recognizing an ideal, deciding
    whether \(\A\) is prime for intersection is in NL.
\end{restatable}

\begin{restatable}[]{theorem}{constructionIntersection}\label{theo:time complexity of intersection decomposition}
	Let \(\A = \Auto\) be an automaton recognizing an ideal.
	Determining whether \(\A\) has an intersection decomposition can be done in time \(\mathcal{O}(|\A||\Sigma|)\),
	and if so, such a decomposition can be constructed in time \(\mathcal{O}(|\A||\Sigma|^2)\).  
\end{restatable}

\noindent
We then turn to union decomposition, and prove a symmetric complexity result.

\begin{restatable}[]{theorem}{constructionUnion}\label{theo:time complexity of union decomposition}
	Given a  minimal automaton  \(\A\) recognizing an ideal, deciding
    whether \(\A\) is prime for union is in NL.
\end{restatable}

\subparagraph*{Related work.}
The general problem of intersection decomposition was introduced in
the seminal work of Orna Kupferman and Jonathan
Mosheiff~\cite{kupferman2015prime}, which defines prime automata,
i.e., automata that are not decomposable, and shows
that deciding whether an automaton is prime for intersection is in EXPTIME.
To date, no better algorithm is known, and the problem is only known to be NP-hard~\cite{DBLP:journals/corr/abs-2605-07031}.
Subsequent work has focused on specific subclasses of automata.
For automata over a single-letter alphabet,
intersection primality is decidable
in NL~\cite{DBLP:conf/mfcs/JeckerKM20},
and for \emph{permutation} automata,
where each input letter induces a bijection
over the state set, intersection primality is decidable
in NP~\cite{jecker2025decomposing}.

Permutation and \emph{aperiodic} automata play complementary roles
in the algebraic study of regular languages:
by the Krohn–Rhodes theorem,
every regular language can be decomposed into a cascade of permutation and aperiodic components.
With the permutation case understood,
the next step is to tackle aperiodic automata.
Prior work in this direction has addressed acyclic automata, providing 
a polynomial time criterion for their intersection decomposability~\cite{spenner2023decomposing}.
Moreover, for \emph{almost-acyclic} automata, that might contain sink states with self loops, the problem is known to be NP-complete~\cite{DBLP:journals/corr/abs-2605-07031}. 
In this paper, we focus on  automata recognizing ideals, another subclass of aperiodic automata.

\subparagraph*{Structure of the paper.}
The paper is organized as follows:
\begin{itemize}
\item
	Section~\ref{sec:preliminaries} recalls standard notions and notations for finite automata and ideals.
\item 
    Section~\ref{sec:decomposition_of_non_linear_automata}
    studies intersection decomposition of \emph{non-linear} automata recognizing
    ideals, i.e., automata whose state sets cannot be totally ordered
    compatibly with the transition function.
    We prove that every automaton of this form is decomposable into  at most \(|\Sigma|\) smaller automata, each
    recognizing an ideal, and that such a decomposition can be
    computed in polynomial time (Theorem~\ref{theo:non-linear are decomposable}).

\item
  Section~\ref{sec:linear_automata_and_decomposition} focuses on
  intersection decomposition of \emph{linear} automata recognizing ideals,
  and provides a primality criterion.
  Moreover, we show that every decomposable automaton of this form
  is decomposable into two smaller automata, each
  recognizing a ideal, and that such a decomposition can be
  computed in polynomial time (Theorem~\ref{prop: decomposition of linear automata}).

\item   
	Section~\ref{sec:union_decomposition} addresses union decompositions,
	and provides a primality criterion (Theorem~\ref{theo:union characterization}).

 \item
 	Finally, Section~\ref{sec: Decomposition in prime automaton} combines the results of the two previous sections
 	to establish that intersection primality of automata recognizing ideals
 	is decidable in NL (Theorem~\ref{theo:space complexity})
 	and computable in polynomial time (Theorem~\ref{theo:time complexity of intersection decomposition}).
\end{itemize}

All results are proven; proofs omitted from the main text are deferred to the appendix.

\section{Preliminaries}\label{sec:preliminaries}

\subparagraph*{Mathematics.}  For all \(n,m \in \N\), we denote by \(\llbracket
n, m \rrbracket\) the set \(\set{k \in \N}{n \leq k \text{ and } k\leq
  m}\).  For every set \(E\) and for every subset \(F\) of \(E\), the
subset \(E\setminus F\) is just denoted \(\overline{F}\) when there is
no ambiguity on \(E\).

\subparagraph*{Ordered relations.} For  every set \(S\), a binary relation
\(\leq\) on \(S\) is a \emph{partial order} if \(\leq\) is
reflexive, transitive and antisymmetric.  A partial order \(\leq\) is
said to be \emph{total} if for all \(q,r \in S\) either \(q \leq r\)
or \(r \leq q\). For an order \(\leq\) on \(S\), an element
\(x\) of some subset \(X \subseteq S\) is a \emph{minimal} [resp. \emph{maximal}] element of
 \(X\) if for all \(y\in X\), \(y\leq x\)
 implies \(x=y\) [resp. \(x \leq y\) implies \(x=y\)].  If \(x\) is a
minimal [resp. maximal] element of \(X\) such that every
\(y\in X\) satisfies \(x\leq y\) [resp.  \(y\leq x\)], then \(x\) is
called the \emph{least element} [resp. \emph{greatest element}] of
\(X\).  Note that if a subset \(X\) has a least [resp. greatest]
element, then it is unique and denoted \(\min X\) [resp. \(\max
  X\)]. Note too that any finite subset \(X\) admits at least one
minimal and one maximal element.

\subparagraph*{Words and languages.}  An \emph{alphabet} is a finite set whose
elements are called \emph{letters}.  A \emph{finite word} over an
alphabet \(\Sigma\) is a finite sequence of letters of \(\Sigma\). The
empty sequence is called the \emph{empty word} and is denoted
\(\varepsilon\). A word \((w(i))_{i \in \setint{n}}\) is classically
written \(w_1 \dots w_n\). Its length \(n\) is denoted \(|w|\). The
empty word is the unique word of length 0.  For every alphabet
\(\Sigma\), the set of all finite words over \(\Sigma\) is denoted
\(\Sigma^*\). Subsets of  \(\Sigma^*\) are called \emph{languages} over \(\Sigma\).

\subparagraph*{Shuffle product.}
Let \(\Sigma\) be an alphabet and let \(u,v \in \Sigma^*\) be
words of respective length \(m\) and \(n\).
The \emph{shuffle} of \(u\) and \(v\), denoted by \(u \shuffle v\), is the
set of words \(w \in \Sigma^*\) of length \(m + n\) such that
there exist two strictly increasing sequences \((i(k))_{k\in
\setint{m}}\) and \((j(k))_{k\in \setint{n}}\) whose images partition \(\setint{m+n}\)
and such that \(u_k = w_{i(k)}\) for all
\(k \in \setint{m}\) and \(v_\ell = w_{j(\ell)}\) for all
\(\ell \in \setint{n}\). This notion extends to
languages in the usual way: for two languages \(L_1,
L_2 \subseteq \Sigma^*\), the \emph{shuffle} of \(L_1\) and \(L_2\) is
the language \(L_1 \shuffle L_2 = \bigcup_{u \in L_1, v \in L_2} u \shuffle v\).
For all \(u, v \in \Sigma^*\), if \(v\in \{u\}\shuffle \Sigma^*\)
we say that \(u\) is a \emph{sub-word} of \(v\) and that \(v\) is an
\emph{upper-word} of \(u\). 
Every language \(L \subseteq \Sigma^*\)  satisfies \(L \subseteq \Shuffle{L}\).
A language is an \emph{ideal} if \(\Shuffle{L} = L\),
or, equivalently, if for all \(u \in L\) every upper-word of \(u\) also belongs to \(L\).

\subparagraph*{Finite automata.}  A \emph{finite state automaton} is a
\(5\)-tuple \(\A = \Auto\) where \(S\) is a finite set of
\emph{states}, \(\Sigma\) is an alphabet, \(\iota \in S\) is the
\emph{initial state}, \(F \subseteq S\) is the set of \emph{accepting
(or final) states} and \(\delta \subseteq S \times \Sigma \times S\)
is the set of \emph{transitions}. If for all \(q \in S\) and for all
\(a \in \Sigma\) there exists at least [resp. at most] one \(q' \in
S\) such that \((q,a,q') \in \delta\), the automaton is said to be
\emph{complete} [resp. \emph{deterministic}]. By abuse of notation,
for a complete deterministic automaton, the unique state \(q' \in S\)
such that \((q,a,q') \in \delta\) is denoted \(\delta(q,a)\).  For
deterministic complete automata, the notation \(\delta\) is extended
to words in the following way: for all \(q \in S\) we set
\(\delta(q,\varepsilon) = q\), and for all \(u \in \Sigma^*\) and
\(a \in \Sigma\) we define inductively \(\delta(q,ua)=
\delta(\delta(q,u),a)\).
Let \(w \in \Sigma^*\) and \(q_1,q^\prime_n \in S\).
A run \(\sigma\) of \(\A\) over \(w\) from \(q_1\)
to \(q^\prime_n\) is a sequence of transitions \((q_i,w_i,q_i')_{i \in
  \setint{n}}\) satisfying \(w_1 \dots w_n= w\) and such that for all
\(i \in \setint{n-1}\), the state \(q_i'\) is equal to \(q_{i+1}\).
The \emph{size} of a deterministic automaton \(\A\), denoted \(|\A|\),
is its number of states. For every automaton \(\A = \Auto\), the
\emph{language recognized} by \(\A\), denoted by \(\lang{\A}\), is the
set of words \(w \in \Sigma^*\) such that \(\delta(\iota,w) \in F\).
Let \(\A = \Auto\) and \(q \in S\).  The \emph{residual} of \(\A\) in
\(q\), denoted \(\Res{\A}{q}\), is the language recognized by the automaton
\((S,\Sigma, q, F, \delta)\) obtained by replacing the initial state of \(\A\) by \(q\).
A language is \emph{regular} if there
exists an automaton that recognizes it. 
A deterministic complete automaton is called minimal if there is no
smaller automaton that recognizes the same language or, equivalently,
if it is the minimal automaton of its language.  In a minimal
automaton \(\A\), if \(\Res{\A}{q}=\Res{\A}{q^\prime}\), then
\(q=q^\prime\).  We say that a state \(q\) of a deterministic complete
automaton is a \emph{sink state} if for all \(a \in \Sigma\) one has
\(\delta(q,a) = q\). For any regular language
\(L\), there exists a unique minimal deterministic complete automaton (up to
state names) recognizing \(L\) called the minimal automaton of \(L\).

\subparagraph*{Accessibility for automata.}  Let \(\A = \Auto\) be a
deterministic complete automaton and let \(q,r \in S\). We say that
\(r\) is \emph{accessible} from \(q\) when there exists \(w \in
\Sigma^*\) such that \(r = \delta(q,w)\). In that case we denote it by
\(q \preccurlyeq_\A r\). When \(q \preccurlyeq_\A r\) and \(q \neq
r\), we write \(q \prec_\A r\).  When either \(q \preccurlyeq_\A r\)
or \(r \preccurlyeq_\A q\) we say that \(q\) and \(r\) are
\emph{comparable}.  When \(\A\) is clear from the context we
simply write \(\preccurlyeq\) and \(\prec\).
A deterministic complete automaton is called \emph{trim}
if every state is accessible from the initial state
and from every state there exists an accessible final state.
In this paper all automata are assumed to be deterministic complete and trim unless otherwise
stated.  Let \(\A\) be a deterministic complete automaton, and let
\(q\) be a state of \(\A\). The set of \emph{ ancestors} of \(q\), denoted
\(\predA{\A}{q}\), is
\(\set{s \in S}{s \preccurlyeq_\A q}\).
The set of \emph{descendants} of \(q\),
denoted
\(\succsA{\A}{q}\), is
\(\set{s \in S}{q \preccurlyeq_\A s}\).
The \emph{family} of \(q\),  denoted \(\mifaA{\A}{q}\),
is \(\pred{q} \cup \succs{q}\).
When \(\A\) is clear from the context we
simply write \(\pred{q}\), \(\succs{q}\) and \(\mifa{q}\).
An automaton \(\A\) is called \emph{linear} if \(\preccurlyeq_\A\) is a total order on its set of
states, otherwise \(\A\) is \emph{non-linear}.
For instance, the automaton on Fig.~\ref{fig:rank} is non-linear since \(q_1\) and \(q_2\) are not
comparable. We denote \(\I\) the set of automata recognizing ideal.

\subparagraph*{Decomposition Problems.}  For a minimal automaton \(\A\), an
\emph{intersection [resp. union] decomposition} of \(\A\) is a finite
sequence of automata \((\A_i)_{i \in \setint{n}}\) such that
\(|\A_i|< |\A|\) for all \(i \in \setint{n}\), and \[\lang{\A} =
\bigcap\limits_{i \in \setint{n}}\lang{\A_i}\quad \text{[resp.}\quad
\lang{\A} = \bigcup\limits_{i \in \setint{n}}\lang{\A_i}]\enspace. \]

  When \(\A\) admits an intersection [resp. union] decomposition, we
  say that \(\A\) is \emph{composite for intersection}
  [resp. \emph{for union}] and otherwise it is \emph{prime for
    intersection} [resp. \emph{for union}].

\subsection{Ideals}
In this section, we recall several useful results on ideals.
Most of them follow either from Higman's Lemma~\cite{higman} (there
is no infinite language of incomparable words for the subword
relation), or from the fact that ideals form a positive variety of
languages~\cite{DBLP:conf/icalp/PinW95}. See
also~\cite{heam2002shuffle,straubing1988partially}.
An example of a minimal automaton accepting an ideal is shown in
Fig.~\ref{fig:rank}.
These classical results can be summarized as follows.

\begin{figure}
\begin{center}
  \begin{tikzpicture}
    \node[state,initial, initial text=,inner sep=3pt, minimum size=0pt,fill=blue!25] (i) at (-2,0) {$\iota$};
    \node[state,inner sep=3pt, minimum size=0pt,fill=blue!25] (0) at (0,0) {$q_0$};
    \node[state,inner sep=3pt, minimum size=0pt,fill=blue!25] (1) at (2,1)  {$q_1$}; 
    \node[state,inner sep=3pt, minimum size=0pt,fill=blue!25] (2) at (2,-1)  {$q_2$}; 
    \node[state,inner sep=3pt, minimum size=0pt,fill=blue!25] (3) at (4,1)  {$q_3$}; 
    \node[state,inner sep=3pt, minimum size=0pt,fill=blue!25] (4) at (4,-1)  {$q_4$}; 
    \node[state,inner sep=3pt, minimum size=0pt,fill=blue!25] (5) at (6,1)  {$q_5$}; 
    \node[state,accepting,inner sep=3pt, minimum size=0pt,fill=blue!25] (6) at
    (6,-1)  {$q_6$}; 

    \path[draw,-latex] (i) edge node[above] {$c$}(0);
    \path[draw,-latex] (i) edge [loop above] node[left] {$a,b$}();
    
    \path[draw,-latex] (0) edge node[above] {$a$}(1);
    \path[draw,-latex] (0) edge node[above] {$b$}(2);
    \path[draw,-latex] (0) edge [loop above] node[left] {$c$}();
    
    \path[draw,-latex] (1) edge  [loop above] node[left] {$a$}();
    \path[draw,-latex] (1) edge  node[above] {$c$}(3);
    \path[draw,-latex] (1) edge[bend left=40]  node[above,pos=0.9] {$b$}(5);

    \path[draw,-latex] (3) edge  [loop above] node[left] {$a$}();
    \path[draw,-latex] (3) edge  node[above] {$b$}(5);
    \path[draw,-latex] (3) edge[]  node[left] {$c$}(4);
    
    \path[draw,-latex] (2) edge  [loop below] node[left] {$b$}();
    \path[draw,-latex] (2) edge  node[above] {$a$}(3);
    \path[draw,-latex] (2) edge[]  node[above] {$c$}(4);

    \path[draw,-latex] (4) edge  [loop below] node[right] {$b,c$}();
    \path[draw,-latex] (4) edge  node[above] {$a$}(6);

    \path[draw,-latex] (5) edge  [loop above] node[left] {$a$}();
    \path[draw,-latex] (5) edge  node[above] {$c$}(4);
    \path[draw,-latex] (5) edge[]  node[right] {$b$}(6);
    
    \path[draw,-latex] (6) edge[loop right]  node[right] {$a,b,c$}(6);

  \end{tikzpicture}
  \end{center}
\caption{Minimal automaton of  \(\{cabb,cacca,cbca\}\shuffle
    \{a,b,c\}^*\).}\label{fig:rank}
\end{figure}

\begin{proposition} \label{prop: hold prop on ideals}
    The following properties hold.
  \begin{enumerate}
  \item \label{theo:ideal are regular}
    All ideals are regular.
  \item \label{lemma: Lmin ideals}\label{lemma:no subwords of Lmin in L}
    If \(L\) is an ideal, then there exists a finite
    language \(K\) such that \(L=K\shuffle \Sigma^*\).
    Moreover, there exists a unique finite language with this property of minimal cardinality,
    which we denote by \(\Lmin\).
    Finally, for any \(u\in \Lmin\), every strict subword
    \(v\) of \(u\) satisfies \(v\notin L\).
  \item \label{lemma:Res of ideal are ideal}
    If  \(\A\) is a deterministic automaton recognizing an ideal,
    then for all \(q \in S, \Res{\A}{q}\) is an ideal.
  \item 
    If  \(\A\) is a minimal automaton recognizing an ideal,
    then \(\A\) is trim and \(\preccurlyeq\) is a partial order.
  \item \label{lemma: ideals unique final state} 
    If \(\A\) is a minimal automaton recognizing an ideal, then it
    has a unique final state, and this final state is a sink state.
  \end{enumerate}
  \end{proposition}

\noindent
We now provide more technical statements that directly follow from these known properties.

\begin{proprep}\label{prop:new prop on ideals}
	Let \(\A = \Auto\) be an automaton recognizing an ideal. The following properties hold.
	\begin{enumerate}
		\item \label{lemma:preccurlyeq and inclusion of residuals}
			Let \(q,r\) be two states
      satisfying \(q \preccurlyeq r\). One
      has \(\Res{\A}{q} \subseteq \Res{\A}{r}\).
		\item \label{lemma:R(delta(a)) subset R(delta(ua))}
			Let \(q \in S\) and \(u,v \in \Sigma^*\). One has \(\Res{\A}{\delta(q,u)} \subseteq \Res{\A}{\delta(q,vu)}\).
		\item \label{lemma: min automata preccurlyeq equality}
			Let \(q \in S\)
			and \(u,v \in \Sigma^*\).
			Then either \(\delta(q,u) \preccurlyeq \delta(q,vu)\),
			or \(\delta(q,u)\) and \(\delta(q,vu)\) are incomparable.
		\item\label{lemma:absorbent implies accepting}
			Let \(q \in S\) be a sink state. If \(\lang{\A} \neq \emptyset\) then \(q \in F\).
	\end{enumerate}
\end{proprep}

\begin{proof}
	\begin{enumerate}
		\item
			By definition of \(\preccurlyeq\), if \(q \preccurlyeq r\),
      then there exists \(u \in \Sigma^*\) such that \(r =
      \delta(q,u)\). For all \(w \in \Res{\A}{q}\), the word \(w\)
      is a sub-word of \(uw\). Hence, since ideals are close by upper-words, we have   \(uw \in
      \Res{\A}{q}\). Now, since \(\A\) is deterministic and
      complete,  \(w \in \Res{\A}{\delta(q,u)}\). Since \(r = \delta(q,u)\), one can conclude that \(w \in \Res{\A}{r}\). Therefore \(\Res{\A}{q} \subseteq \Res{\A}{r}\).
    \item 
    	Let \(w \in \Res{\A}{\delta(q,u)}\). Since \(\A\) is
			deterministic, \(uw \in \Res{\A}{q}\). Since \(\Res{\A}{q}\)
			is an ideal, and since \(uw\) is a sub-word of \(vuw\),
		 	we have  \(vuw \in \Res{\A}{q}\).
 			Therefore, \(w \in \Res{\A}{\delta(q,vu)}\)
 			as \(\A\) is deterministic.
 		\item 
 			If \(\delta(q,vu) \not\preccurlyeq \delta(q,u)\),
			the statement is immediately satisfied
			as either \(\delta(q,u) \preccurlyeq \delta(q,vu)\)
			or \(\delta(q,u) \not\preccurlyeq \delta(q,vu)\),
			and the later case implies that
			\(\delta(q,u)\) and \(\delta(q,vu)\) are incomparable.
	
			Now in the case where \(\delta(q,vu) \preccurlyeq \delta(q,u)\),
			on the one hand Item~\ref{lemma:preccurlyeq and inclusion of residuals} of this property yields that
			\(\Res{\A}{\delta(q,vu)} \subseteq \Res{\A}{\delta(q,u)}\),
			and on the other hand 
			Item~\ref{lemma:R(delta(a)) subset R(delta(ua))} of this property
			yields that \(\Res{\A}{\delta(q,u)} \subseteq \Res{\A}{\delta(q,vu)}\).
			Therefore \(\Res{\A}{\delta(q,u)} = \Res{\A}{\delta(q,vu)}\)
			Hence, as \(\A\) is minimal,
			\(\delta(q,u) = \delta(q,vu)\).
			In particular \(\delta(q,u) \preccurlyeq \delta(q,vu)\).
		\item
			Since \(\A\) is trim, there exists \(w \in \Sigma^*\) such
			that \(\delta(\iota,w)=q\). Let \(v \in\lang{\A}\). Since \(\lang{\A}\)
			is an ideal, one has \(wv \in \lang{\A}\). And since \(q\) is a sink
			state, \(\delta(\iota,vw)= q\). Those two statements implies
			that \(q \in F\).\qedhere
	\end{enumerate}
\end{proof}

We will denote by \(\I\) the class of minimal automata recognizing ideals.

\section{Intersection-decomposition of non-linear automata}\label{sec:decomposition_of_non_linear_automata}

In this section, decomposition and prime automata
are understood with respect to intersection.
We show that every non-linear automaton recognizing an ideal is decomposable.

\begin{restatable}[]{theorem}{decompOfNonLinear}
  \label{theo:non-linear are decomposable}
  		Let \(\A\) be a non-linear minimal automaton
  		over an alphabet \(\Sigma\) that recognizes an ideal.
  		Then \(\A\) has a decomposition into at most \(|\Sigma|\) automata, each recognizing an ideal.
\end{restatable}

  \noindent
  The proof follows from the results established in the following sections.
  In Section~\ref{sec:qsep}, we define a set of states, the \emph{separator set} \(\setsep\),
  and establish its main properties.
  In Section~\ref{sec:decomposition}, we associate with each state \(\rho \in \setsep\)
  a \emph{family automaton} \(\A(\rho)\), and prove that these automata form a decomposition of \(\A\).
  More precisely, for each \(\rho \in \setsep\) the automaton \(\A(\rho)\) is strictly smaller than \(\A\)
  (Proposition~\ref{lemma:|A(a)| < |A|}) and recognizes an ideal
  (Proposition~\ref{lemma:RA'q' equality}).
  Moreover, 
  \(\lang{\A} = \bigcap_{\rho \in \setsep} \lang{\A(\rho)}\)
  (Proposition~\ref{lemma:L(A) = inter L(A(a))}).
  The bound on the size of the decomposition follows from the fact that
  \(|\setsep| \leq |\Sigma|\)
  (Proposition~\ref{prop:sep}.\ref{lemma:transition from sep to setsep}).

\subsection{Ranks and Separation States}\label{sec:qsep}
We introduce in this section the key notion of \emph{separator state} \(\sep\) of a non-linear automaton \(\A\).
Intuitively, this state is the first point in \(\A\) at which a branching towards two incomparable states occurs.
We then define the notion of \emph{rank} on the states of \(\A\).
This notion provides an alternative characterization of the separator
state (Lemma~\ref{lemma:unicity above sep}),
and also allows us to introduce an important ingredient of our
construction: the \emph{separator set} \(\setsep\), which consists of the states
immediately following the separator state.

\subparagraph*{Separator state.}
We characterize the state in non-linear automata just before the first two states that are not comparable.
To introduce this notion formally,
we first establish a structural property on non-linear automata
recognizing ideals.

\begin{proprep}\label{lemma:carac_non-linear}
	Let \(\A = \minAuto\) be a minimal automaton accepting an ideal.
	Then \(\A\) is non-linear if and only if
	the subset of states 
	\[\set{q \in
		S\setminus F}{\succs{q}\setminus\{q\} \text{ has no least element for } \preccurlyeq}\]
	is non-empty and has a least element for \(\preccurlyeq\).
\end{proprep}

\begin{proof}
	Let \(H_{\A}\) denote the set \(\set{q \in
		S\setminus F}{\succs{q}\setminus\{q\} \text{ has no least element for } \preccurlyeq}\).
	First, we show that \(H_{\A}\) is non-empty if and only if \(\A\) is non-linear.
	
	Suppose that  \(H_{\A}\) is non-empty.
	Then there exists a non final state \(q\) such that the subset
	\(\succs{q}\setminus\{q\}\) has no least element for \(\preccurlyeq\).
	Since \(\A\) is trim and \(q\) is not a final state, the set \(\succs{q}\setminus\{q\}\) is non-empty.
	The absence of a least element implies that it contains at least two incomparable minimal elements.
	Therefore the order \(\preccurlyeq\) is not total, and \(\A\) is non-linear.
	
	Conversely, assume that \(H_{\A}\) is empty.
	Then for every non final state \(q\) the subset
	\(\succs{q}\setminus\{q\}\) has a least element for \(\preccurlyeq\).
	Let \(q_0=\iota\) be the initial state of \(\A\), and for each non
	final state \(q_i\) let \(q_{i+1}\) be the least element  of
	\(\succs{q_i}\setminus\{q_i\}\).
	By construction, \(q_i\preccurlyeq q_{i+1}\).
	Moreover, since \(\A\) is trim, it is easy to show by induction that \(\pred{q_i} = \set{q_j\in S}{j \leq i}\). Hence every state \(q\) of \(\A\) is visited in this chain.
	It follows that \(\preccurlyeq\) is a total order, hence \(\A\) is linear.
	
	To prove the second part of the statement,
	we now show that if \(H_\A\) is non-empty then it has a
	least element.
	Towards building a contradiction, suppose there exist two distinct element of \(H_\A\), denoted \(q_1\) and \(q_2\).
	Let \(T=\{q \mid q \preccurlyeq q_1 \text{ and } q\preccurlyeq q_2\}\).
	Since the initial state is in \(T\),
	it is non-empty. Let \(p\) be a maximal
	element of \(T\).
	Then \(p\preccurlyeq q_1\) and \(p\preccurlyeq q_2\).
	Since \(q_1\) is minimal in \(H\), we must have \(p\notin H_\A\).
	Consequently, \(\succs{p}\setminus\{p\}\) has a least element
	\(r\). Since \(q_1\in \succs{p}\setminus\{p\}\) and \(q_2\in
	\succs{p}\setminus\{p\}\), one has \(r\preccurlyeq q_1\) and \(r\preccurlyeq
	q_2\), hence \(r\in T\).
	Since \(p\preccurlyeq r\) and \(p \neq r\), this contradicts the maximality of \(p\)
	in \(T\), concluding the proof.
\end{proof}

Now for every non-linear automaton \(\A = \Auto\) that
recognizes an ideal, the \emph{separator state} of \(\A\),
denoted \(\sep\), is defined by
\[\sep = \text{min}_\preccurlyeq \set{q \in S}{\succs{q}\setminus\{q\} \text{ has no least element
            for} \preccurlyeq}.\]
In other words, \(\sep\) is the least element (for
\(\preccurlyeq\)) whose strict successors do not admit a least element.
This notion is illustrated in Fig.~\ref{fig: nonlinear}: the set
\(\set{\sep \in S}{\succs{\sep}\setminus\{\sep\}}\) has two distinct
  minimal elements \(\rho_1\) and \(\rho_2\).

\begin{figure}

       \includegraphics[width=0.9\textwidth]{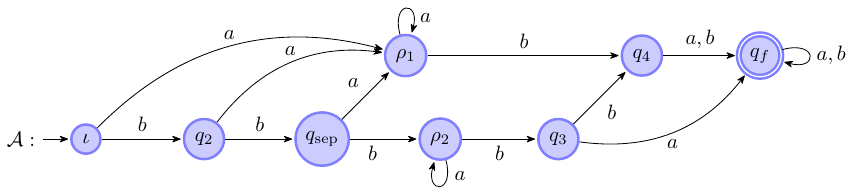}

    \caption{A non-linear automaton recognizing an ideal.}
    \label{fig: nonlinear}
\end{figure}

\subparagraph*{Rank.} The notion of rank is a classical one in directed acyclic graphs and
corresponds intuitively to a depth notion in the graph. This notion represents how far a state can be from the initial state. More formally, let \(\A = \minAuto\) be in \(\I\) and let
  \(q \in S\). The \emph{rank} of \(q\), denote \(\norm{q}\), is the
  size of a longest path in \(\A\) from \(\iota\) to \(q\) that
  doesn't visit twice any state
  We denote by \(\setsep\) the set of states of the lowest rank held by at least two states.
  It has the property that there exists a transition from \(\sep\)
  to each state of \(\setsep\).

For instance, on
illustrating Fig.\ref{fig:rank}, we have
\(\norm{\iota}=0\), \(\norm{q_2} = 1\), \(\norm{\sep} = 2\), \(\norm{\rho_1} = \norm{\rho_2} = 3\), \(\norm{q_3} = 4\), \(\norm{q_4} = 5\) and \(\norm{q_f} = 6\).

The main properties linking ranks and states are summed up in Proposition~\ref{porp:rank and states}.

\begin{proprep}\label{porp:rank and states}
  Let \(\A = \minAuto\) be in \(\I\). The following statements holds.
  \begin{enumerate}
     \item \label{lemma:norm and transition}\label{cor:norm and path} 
     For all \(q,s \in S\) verifying \(q \prec s\), one has \(\norm{q} < \norm{s}\). 
     \item \label{cor:eq norm and preccurlyeq}
       For all \(q,s \in S\). If \(\norm{q} = \norm{s}\) and \(q \neq s\) then \(q\) and \(s\) are incomparable for \(\preccurlyeq\).
    \item \label{lemma:lower than n implis path taht uses n}
       Let \(q \in S\) and \(n \in\mathbb{N}\) such that \(n \leq \norm{q}\). There exists \(r \in \pred{q}\) of rank \(n\).
   \end{enumerate} 
\end{proprep}

\begin{proof}
  \begin{enumerate}
    \item
      First we show that if there is \(a \in \Sigma\) such that \(\delta(q,a)= s\), then \(\norm{q} < \norm{s}\). Then we generalize this result to obtain the wanted statement. Let \(\pi\) be a loopfree path from \(\iota\) to \(q\) of length
      \(\norm{q}\). Since \(\preccurlyeq\) is an order relation, and since
      \(q\leq s\), \(\pi\) doesn't visit \(s\). Therefore the path
      \(\pi\cdot(q,a,s)\) is a loopfree path from \(\iota\) to \(s\) of length
      \(\norm{q}+1\), proving that \(\norm{s}\geq \norm{q}+1 > \norm{q}\). The
      general result is obtained by induction on a run from \(q\) to \(s\).
    \item 
      By a direct contraposition of the first statement, one has the second one.
    \item
      Let \(\pi\) be a loopfree path from \(\iota\) to \(q\).  Let \(w =
      w_1 \dots w_{\norm{q}}\) be the label of \(\pi\) and let \(s =
      \delta(\iota,w_1 \dots w_n)\). By definition of \(\preccurlyeq\) one
      has \(s \preccurlyeq q\). Then it suffices to show that \(\norm{s} =
      n\). Since \(\delta(\iota,w_1\dots w_n) = s\), one has \(\norm{s}
      \geq n\). Now, for all words \(w'\) such that \(\delta(\iota,w') = s\), one
      has \(\delta(\iota,w'w_{n+1}\dots w_k) = q\). Therefore, by definition of
      \(\norm{q}\), the word \(w'w_{n+1}\dots w_k\) has a length smaller or
      equal to \(\norm{q}\). It follows that the length of \(w'\) is at most
      \(\norm{q}-\norm{q} + n=n\). Hence \(\norm{s}
      \leq n\), proving the lemma. \qedhere
  \end{enumerate}
\end{proof}

These statements allow us to define another characterization of \(\sep\).

\begin{lemrep}\label{lemma:unicity above sep}
	Let \(\A = \minAuto\) be a non-linear automaton in \(\I\) and
        for all \(m \in \N\) let \(a_m = |\set{s \in S}{\norm{s} =
          m}|\). Then \(\sep\) is the only state of rank \(\max\set{n \in \N}{ \forall m
          \leq n, a_m = 1}.\)
\end{lemrep}

\begin{proof}
  Note first that \(\iota\) is the unique state of rank \(0\) and
  \(a_0=1\). Therefore the set \(\max\set{n \in \N}{ \forall m \leq n,
    a_m = 1}\) is non-empty. Let \(n_0\) be its maximum. By
  definition, \(s_{n_0}=1\). Let \(q\) be the unique element of rank
  \(n_0\). Let's show that \(q = \sep\). This proof is done in two
  steps. First we will show that \(\sep \preccurlyeq q\) and next
  that \(q = \sep\).

  \begin{itemize}
    \item Note that \(q\neq q_f\). Otherwise, \(\norm{q_f}=|\A|\),
      implying that there is a loopfree path in \(\A\) visiting all
      states and \(\A\) would be linear.

      Therefore, by definition of \(q\), one has \(s_{\norm{q}+1}\geq 2\).
      Let \(r_1,r_2 \in S\) be two states verifying that \(\norm{r_1} = \norm{r_2} =
      \norm{q}+1\). By Proposition~\ref{porp:rank and states}.\ref{lemma:lower than n implis path taht
        uses n} and since \(q\) is the only state of rank
      \(\norm{q}\), one has \(q \preccurlyeq r_1\), \(q \preccurlyeq
      r_2\), and \(r_1\) and \(r_2\) are not comparable. 

      Now let \(h\in \succs{q}\setminus\{q\}\) satisfying \(h\prec r_1\).
      One has \(\norm{h} < \norm{r_1}=\norm{q}+1\) by
      Proposition~\ref{porp:rank and states}.\ref{lemma:norm and transition}. Moreover, since \(h\in
      \succs{q}\), we also have \(\norm{h} \geq \norm{q}\). It follows
      that \(\norm{h}=\norm{q}\). Since \(s_{\norm{q}}=1\),
      \(h=q\). But \(h\in \succs{q}\setminus\{q\}\), a
      contradiction. Therefore, \(r_1\) is minimal in
      \(\succs{q}\setminus\{q\}\). Similarly \(r_2\) is minimal in
      \(\succs{q}\setminus\{q\}\). Consequently, \(\succs{q}\setminus\{q\}\)
      has no least element. By definition of \(\sep\), one has
      \(\sep \preccurlyeq q\).

    \item Now, assume that \(\sep \prec q\). In this case
      \(\norm{\sep}< \norm{q}\). Hence,
      \(s_{\norm{\sep}}=s_{\norm{\sep}+1}=1\).  Let \(u\) be the
      unique state of rank \(\norm{\sep}+1\) and \(v\) be a minimal
      element of \(\succs{\sep}\setminus{\sep}\). One has \(\norm{v} >
      \norm{\sep}\).  By Proposition~\ref{porp:rank and states}.\ref{lemma:lower than n implis path
        taht uses n}, \(u\) is in \(\pred{v}\). By minimality of
      \(v\), one has \(u=v\). Therefore, \(u\) is the unique minimal
      element of \(\succs{\sep}\setminus{\sep}\), a contradiction. 
  \end{itemize}
  It follows that \(\sep = q\).
\end{proof}

\subparagraph*{Separator set.}	The rank characterization of the separator state makes it natural to consider the set of states of rank \(\norm{\sep}+1\). The decomposition we construct contains one automaton for each state in this set. We also introduce the set of letters leading from \(\sep\) to the separator set. More formally, let \(\A = \minAuto\) be a non-linear automaton in \(\I\). The \emph{separator set}, denoted \(\setsep\), is the subset of \(S\) defined by \(\setsep = \set{q \in S}{\norm{q} = \norm{\sep}+1}\).

Finally we can define the technical properties of \(\sep\) and  \(\setsep\). 

\begin{proprep}\label{prop:sep}
  Let \(\A = \minAuto\) be a non-linear minimal automaton recognizing an ideal.
  \begin{enumerate}
  \item \label{lemma: above sep mifa = S}
    If \(q \in S\) satisfies \(q \preccurlyeq \sep\), then \(\mifa{q}
    = S\).

  \item \label{lemma:pred(q in setsep) fully ordered}
    If \(\rho \in \setsep\), then \(\preccurlyeq\) induces a total order
    on \(\pred{\rho}\setminus \{\rho\}\).\

  \item \label{lemma:transition from sep to setsep} Let \(\rho \in
    \setsep\). There exists \(a \in \Sigma\) such that
    \(\delta(\sep,a) = \rho\).
     
   \item \label{lemma:size setsep geq 2} \(2 \leq |\setsep| \leq |\Sigma|\).
  \end{enumerate}
\end{proprep}

\begin{proof}
\begin{enumerate}
\item By Proposition~\ref{porp:rank and states}.\ref{cor:norm and path}, \(\norm{q} \leq
  \norm{\sep}\). Let \(s \in S\). Two cases arise.
   \begin{itemize}
     \item Suppose that \(\norm{s} \leq \norm{q}\). Then Proposition~\ref{porp:rank and states}.\ref{lemma:lower than n
     	implis path taht uses n} implies that there exists  \(s' \in \pred{q}\) of rank \(\norm{s}\).
     	However, since \(\norm{\sep} \geq \norm{q} \geq \norm{s}\), Lemma~\ref{lemma:unicity above sep} implies
       that \(s\) is the only state of rank \(\norm{s}\).
       Therefore \(s = s' \in  \pred{q}\).

     \item Suppose that \(\norm{q} \leq \norm{s}\). Then Proposition~\ref{porp:rank and states}.\ref{lemma:lower than n
     	implis path taht uses n} implies that there exists  \(q' \in \pred{s}\) of rank \(\norm{q}\).
     However, since \(\norm{\sep} \geq \norm{q}\), Lemma~\ref{lemma:unicity above sep} implies
     that \(q\) is the only state of rank \(\norm{q}\).
     Therefore \(q = q' \in  \pred{s}\).
	\end{itemize}
     Thus by definition of family one has \(\mifa{q} = S\).

\item Let \(q,s \in \pred{\rho}\setminus \{\rho\}\), using the previous result \(q \in \mifa{s}\) hence \(q\) and \(s\) are comparable.

 \item Let \(u = u_1\dots u_k\) be a word such that the run of \(\A\)
 over \(u\) from \(\iota\) leads to \(\rho\) without cycle.
 By definition of the
   rank \(\norm{\delta(\iota,u_1 \dots u_{k-1})} =
   \norm{\sep}\). Since \(\sep\) is the only state of rank
   \(\norm{\sep}\), by Lemma~\ref{lemma:unicity above sep},
   \(\delta(\iota,u_1\dots u_{k-1}) = \sep\). Thus \(\delta(\sep,u_k)
   = \rho\). 
   
   \item
   Since \(\A\) is deterministic, the fact that \(|\sep| \leq |\Sigma|\) follows immediately from Item~\ref{lemma:transition from sep to setsep}.
   Moreover, \(|\setsep| \neq 0\) as otherwise \(\preccurlyeq\) would be a total order on
   \(S\), contradicting the non-linearity of \(\A\).
   Finally, \(|\setsep| \neq 1\) since \(\sep = \max\set{n \in
   	\N}{ \forall m \leq n, a_m = 1}\) where \(a_m = |\set{s \in
   	S}{\norm{s} = m}|\) by Lemma~\ref{lemma:unicity above sep}.
   Therefore \(2 \leq |\setsep| \leq |\Sigma|\). \qedhere
  \end{enumerate}
\end{proof}

\subsection{Decomposition of non-linear Automata}\label{sec:decomposition}

\begin{figure}[t]
  \centering
  \includegraphics[width=1\textwidth]{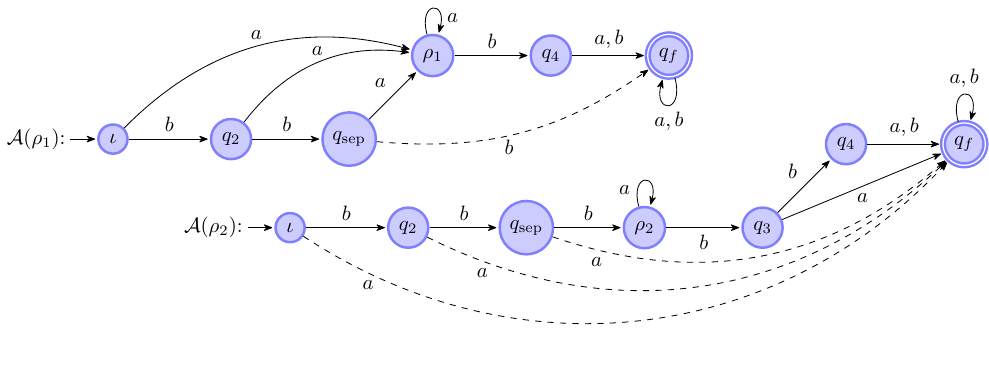}
  
  \caption{Naive decomposition of the automaton \(\A\) depicted in Fig.~\ref{fig: nonlinear}
    into automata that don't recognize ideals.
    The transitions modified compared to the initial automaton are dashed.}
  \label{fig: naive_decompositionAB}
\end{figure}

\begin{figure}[t]
	\centering
	\includegraphics[width=1\textwidth]{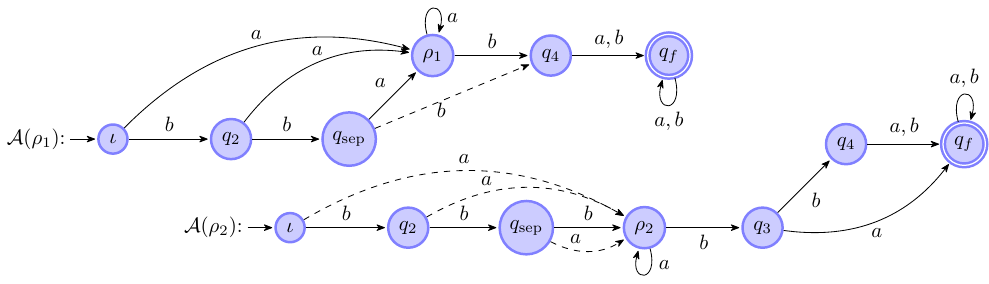}
	
	\caption{Decomposition of the automaton \(\A\) depicted in Fig.~\ref{fig: nonlinear}
		into its family automata \(\A(\rho_1)\) and \(\A(\rho_2)\).
		The transitions modified compared to the initial automaton are dashed.}
	\label{fig: decompositionAB}
\end{figure}

Let us first describe intuitively the proposed construction to
decompose non-linear automata. For every state \(\rho\) of \(\setsep\)
we define a \emph{family automaton}
\(\A(\rho)\) whose set of states is \(\mifa{\rho}\).
The main idea is to redirect into \(\mifa{\rho}\) all the transitions of \(\A\) that leave this set,
while preserving the property of recognizing an ideal.
A naive approach, illustrated on Fig.~\ref{fig: naive_decompositionAB}, would be to redirect such transitions to the unique final sink state.
Although this yields a decomposition of \(\A\), it does not, in general, preserve automata recognizing ideals.
Ensuring that the decomposition remains within this class therefore requires a
more careful construction.
The two main results underlying Theorem~\ref{theo:non-linear are decomposable} are
Proposition~\ref{lemma:RA'q' equality}, which shows that each family automaton \(\A(\rho)\)
recognizes an ideal, and Proposition~\ref{lemma:L(A) = inter L(A(a))},
which ensures that the family automata \((\A(\rho))_{\rho \in \setsep}\)
provide a valid decomposition of \(\A\).

\subparagraph*{Family automaton.} Let \(\A = \minAuto\) be a non-linear minimal automaton that
  recognizes an ideal. For all \(\rho \in \setsep\), we define the \emph{family automaton}
  \(\A(\rho) = (\mifa{\rho},
  \Sigma, \iota, F, \delta_\rho)\),
  where the set of transitions \(\delta_\rho\) is defined as follows.
  For all \(q \in
  \mifa{\rho}\) and \(a \in \Sigma\), we let 
  \[\delta_\rho(q,a) = \delta(s^\rho_{(q,a)},a) \text{, where } 
  s_{(q,a)}^\rho= \min\set{s \in \mifa{\rho}}{q \preccurlyeq s \wedge
  \delta(s,a) \in \mifa{\rho}}.\]
The existence of the least element is provided by Lemma~\ref{lemma:existence of s_(q,b)} which is deferred to appendix.
This construction is illustrated in Fig.~\ref{fig: nonlinear} and~\ref{fig: decompositionAB}:
the automaton depicted in Fig.~\ref{fig: nonlinear} is
decomposed into the two automata \(\A_{\rho_1}\) and \(\A_{\rho_2}\) depicted in Fig.~\ref{fig: decompositionAB}.
Note that the set of states of each automaton \(\A(\rho_i)\) is a
strict subset of the sets of \(\A\), providing the size constraint for
the decomposition. All the transitions starting from a removed state
are deleted. Furthermore, \(\delta_{\rho_i}\) and \(\delta\) coincide on transitions whose source and target
both belong to \(\mifa{\rho_i}\).
The two sets of transitions critically differ for transitions
starting from a state of \(\mifa{\rho}\) and leaving this set:
in this case, the transition set \(\delta_{\rho_i}\) redirects the transition to
the least state greater than or equal to \(\rho_i\) pointed by
the same letter in \(\delta\).
We show that the family automata are smaller than \(\A\).

\begin{proposition}\label{lemma:|A(a)| < |A|}
    Let \(\A\) be a non-linear automaton that recognizes an ideal.
    Then, for every \(\rho \in \setsep\), one has \(|\A(\rho)|<|\A|\).
\end{proposition}

\begin{proof}
  Proposition~\ref{prop:sep}.\ref{lemma:size setsep geq
  2} states that \(|\setsep| \geq 2\),
  hence there exists \(q \in \setsep\setminus\{\rho\}\).
  Then, by
  definition of \(\setsep\), \(\norm{q} =
  \norm{\rho}\),
  and Proposition~\ref{porp:rank and states}.\ref{cor:eq norm and preccurlyeq} implies \(q \notin \mifa{\rho}\).
  Consequently, the construction of \(\A(\rho)\)
  excludes at least one state of \(|\A|\),
  namely \(q\), hence \(|\A(\rho)|< |\A|\).
\end{proof}

\begin{lemrep}\label{lemma:existence of s_(q,b)}
 Let \(\A = \minAuto\) be in \(\I\). Let \(\rho \in \setsep\), \(q \in \mifa{\rho}\) and 
 \(a \in \Sigma\). The set \(\set{s \in \mifa{\rho}}{q \preccurlyeq s
   \wedge \delta(s,a) \in \mifa{\rho}}\) has a least element.
\end{lemrep}

\begin{proof}
If \(\rho \preccurlyeq q\) then \(q\) is the minimum of this set. And
if \(q \preccurlyeq \rho\), then \(\rho\) belongs to this set, and
Proposition~\ref{prop:sep}.\ref{lemma:pred(q in setsep) fully ordered} gives that \(\pred{\rho}\) is a totally ordered finite set,
ensuring the existence of the least element.
\end{proof}

\begin{toappendix}
Before proving results on \(\A(\rho)\), we first point out in Lemma~\ref{lemma: delta' above q_sep} and~\ref{lemma:A' when q leq q_a} as well as in
Corollary~\ref{cor:A' when q leq q_a} some relations between the transitions of
\(\A\) and the transitions of  \(\A(\rho)\).

\begin{lemma}\label{lemma: delta' above q_sep}
   Let \(\A = \minAuto\) be a non-linear minimal automaton that
   recognizes an ideal, and let \(\rho \in \setsep, a \in \Sigma, q \in
   \mifa{\rho}\). If \(\delta_\rho(q,a) \prec_\A
   \rho\), then \(\delta_\rho(q,a) = \delta(q,a)\).
\end{lemma}

\begin{proof}
  Let \(s = s_{q,a}^\rho\). By definition of \(\delta_\rho\) one has that \(\delta_\rho(q,a) = \delta(s,a)\), therefore it suffices to show that \(s = q\) to conclude the lemma. By hypothesis \(\delta_\rho(q,a) \prec_\A \rho\), hence Proposition~\ref{prop:sep}.\ref{lemma: above sep mifa = S} says that \(\delta(s,a)\) and \(\delta(q,a)\) are comparable in \(\A\). Also by definition of \(s\), one has \(q \preccurlyeq_\A s\), hence Proposition~\ref{prop:new prop on ideals}.\ref{lemma: min automata preccurlyeq equality} yields \(\delta(q,a) \preccurlyeq \delta(s,a)\). Thus by transitivity, since \(\delta(s,a) \prec_\A \rho\), one has \(\delta(q,a) \prec_\A \rho\). This says that \(\delta(q,a) \in \mifa{\rho}\), which forces by definition of \(s_{q,a}^\rho\) that \(s = q\). Hence one can concludes that \(\delta_\rho(q,a) = \delta(q,a)\).
\end{proof}

\begin{lemma}\label{lemma:A' when q leq q_a}
	Let \(\A = \minAuto\) be a non-linear minimal automaton that recognizes an ideal, and let \(\rho \in \setsep, q \in \mifa{\rho}\). For all \(a \in \Sigma\) such that \(\delta(q,a) \in \mifa{\rho}\), one has \(\delta(q,a)=\delta_\rho(q,a)\).
\end{lemma}

\begin{proof}
   Let \(a \in \Sigma\) such that \(\delta(q,a) \in
   \mifa{\rho}\). Then \(q\) is exactly the least element
   \(\min\set{s' \in \mifa{\rho}}{s' \preccurlyeq q' \wedge
     \delta(s',a) \in \mifa{\rho}}\).
\end{proof}

As a direct consequence, one gets the following corollary.

\begin{corollary}\label{cor:A' when q leq q_a}
	Let \(\A = \Auto\) be a non-linear minimal automaton that recognizes an ideal, and let \(\rho \in \setsep, q \in \succs{\rho}\). The accessible part of the automaton \((S,\Sigma, q, F, \delta)\) is equal to the accessible part of the automaton \((\mifa{\rho},\Sigma, q, F, \delta_\rho)\).
\end{corollary}

The next lemma is dedicated to establish that for all \(q, q' \in
\mifa{\rho}\) one has \(q \preccurlyeq_\A q'\) if only if \(q
\preccurlyeq_{\A(\rho)} q'\). The proof is based on the fact that the
modified transitions in \(\A(\rho)\) can be seen as shortcuts.

\begin{lemma}\label{lemma: preccurlyeq in A' and A}\label{lemma: preccurlyeq in A and A'}
	Let \(\A = \minAuto\) be a non-linear minimal automaton that recognizes an ideal and \(\rho \in \setsep\). The relation \(\preccurlyeq_{\A(\rho)}\) is a partial order and for all \(q, q' \in \mifa{\rho}\) one has \(q \preccurlyeq_\A q'\) if only if \(q \preccurlyeq_{\A(\rho)} q'\).
\end{lemma}

\begin{proof}
  Let \(q,q^\prime \in \mifa{\rho}\). 
  \begin{itemize}
  \item We first prove that if \(q \preccurlyeq_{\A(\rho)} q^\prime\), then
    \(q \preccurlyeq_{\A} q^\prime\).

    Since \(q \preccurlyeq_{\A(\rho)} q^\prime\), there exists \(u \in
    \Sigma^*\) such that \(\delta_\rho(q,u) = q^\prime\). Let's show by
    induction on \(|u|\) that \(q \preccurlyeq_\A q^\prime\).

  If \(|u|=0\), then \(u = \varepsilon\) and, therefore \(q =
  q^\prime\). Hence \(q \preccurlyeq_\A q^\prime\).

 Now, let \(u'\in \Sigma^*\) and \(a \in \Sigma\) be such that \(u = au'\). One has \(\delta_\rho(\delta_\rho(q,a),u') = q^\prime\), and \(|u'| < |u|\).  By induction, we have \(\delta_\rho(q,a) \preccurlyeq_\A q^\prime\). Since \(\preccurlyeq_\A\) is transitive, it  suffices to show that \(q \preccurlyeq_\A \delta_\rho(q,a)\).

 If \(\delta(q,a) \in \mifa{\rho}\), then by definition of \(\delta_\rho\), we have \(\delta_\rho(q,a)=\delta(q,a)\). Then, by definition of \(\preccurlyeq_\A\), one has \(q \preccurlyeq_\A \delta_\rho(q,a)\). If \(\delta(q^\prime,a) \notin \mifa{\rho}\), then, by definition of \(\delta_\rho\) there exists \(r \in S\) such that \(q \preccurlyeq_A r\) and that \(\delta_\rho(q^\prime,a) = \delta(r,a)\). By transitivity of \(\preccurlyeq_\A\) one has \(q \preccurlyeq_\A \delta_\rho(q,a)\).

\item We now show that \(\preccurlyeq_{\A(\rho)}\) is an order. By
  definition of \(\preccurlyeq\), the relation \(\preccurlyeq_{\A(\rho)}\)
  is reflexive and transitive. Then we only have to show that
  \(\preccurlyeq_{\A(\rho)}\) is antisymmetric.

  Let \(p,p^\prime \in \mifa{\rho}\) be such that \(p^\prime
  \preccurlyeq_{\A(\rho)} p\) and \(p \preccurlyeq_{\A(\rho)} p^\prime\). The
  argument above proves that \(p \preccurlyeq_{A} p^\prime\) and
  \(p^\prime \preccurlyeq_{\A} p\). Finally the antisymmetry of
  \(\preccurlyeq_{\A}\) provides that \(p = p^\prime\). Thus
  \(\preccurlyeq_{\A(\rho)}\) is antisymmetric.

\item  We prove now that if \(q \preccurlyeq_{\A} q'\), then \(q \preccurlyeq_{\A(\rho)} q'\).

This proof is done by case distinction
on the relations between \(q\), \(q'\) and \(\rho\).
   \begin{itemize}
     \item If \(\rho \preccurlyeq_{\A} q\):

      Let \(u \in \Sigma^*\) such that \(\delta(q,u) = q'\),
      Corollary~\ref{cor:A' when q leq q_a} says that \(\delta_\rho(q,u) =
      q'\). Hence, \(q \preccurlyeq_{\A(\rho)} q'\).

     \item If \(q' \preccurlyeq_{\A} \rho\):

      Let \(u_1\dots u_n \in \Sigma^*\) be such that \(\delta(q,u_1 \dots u_n) = q'\). By transitivity of \(\preccurlyeq_\A\), for all \(i \in \setint{n}\), one has \(\delta(q,u_1 \dots u_i) \preccurlyeq_A \rho\). Consequently, \(\delta(q,u_1 \dots u_i) \in \mifa{\rho}\). Hence Lemma~\ref{lemma:A' when q leq q_a} applied \(n\) times provides that \(\delta_\rho(q,u_1 \dots u_n) = \delta(q,u_1 \dots u_n)\). Therefore \(\delta_\rho(q,u_1 \dots u_n) = q'\), providing by definition of \(\preccurlyeq\) that \(q \preccurlyeq_{\A(\rho)} q'\).

   \item If \(q \preccurlyeq_{\A} \rho\) and \(\rho \preccurlyeq_{\A} q'\):

     The first case of this part of the proof applied on \(q \preccurlyeq_{\A} \rho\) gives \(q \preccurlyeq_{\A(\rho)} \rho\), and the second case applied on \(\rho \preccurlyeq_{\A} q'\) provides \(\rho \preccurlyeq_{\A(\rho)} q'\). Then the transitivity of \(\preccurlyeq_{A(\rho)}\) allows us to conclude \(q \preccurlyeq_{\A(\rho)} q'\)
	\end{itemize}

 \end{itemize}
\end{proof}

\begin{lemma}\label{lemma: E = E'}
Let \(\A = \minAuto\) be a non-linear minimal automaton that
recognizes an ideal, let \(\rho \in \setsep\) and let \(q \in
\mifa{\rho}\). One has \(\set{a \in \Sigma}{\delta(q,a) =
  q}=\set{a \in \Sigma}{\delta_\rho(q,a) = q}\).
\end{lemma}

\begin{proof}
  Let \(q\) be a state of \(\A(\rho)\). We denote \(E' = \set{b \in
    \Sigma}{\delta_\rho(q,b) = q}\) and we also write \(E~=~\set{a \in \Sigma}{\delta(q,a)
    = q}\).~Since \(q \in \mifa{\rho}\), \(q\) and \(\rho\) are comparable: we reason by case distinction on the
  relation between \(\rho\) and \(q\).

  \begin{itemize}
  \item If \(\rho \preccurlyeq_{\A} q\): by Lemma~\ref{lemma:A' when q
   leq q_a}, for all \(a \in \Sigma, \delta_\rho(a,b) = \delta(a,b)\),
   hence \(E = E'\).
	
 \item Otherwise if \(q \prec \rho\): Let \(a \in E'\). Since \(q\) is a state of  \(A(\rho)\), Lemma~\ref{lemma: delta' above q_sep} provides that
   \(\delta_\rho(q,a) = \delta(q,a)\). Thus, since \(\delta_\rho(q,a) = q\), one
   has \(\delta(q,a) = q\). Therefore \(a \in E\) and \(E' \subseteq
   E\).

 Let \(a \in E\). Therefore \(\delta(q,a) \in \mifa{\rho}\)
 since \(\delta(q,a) = q\). Thus, Lemma~\ref{lemma:A' when q leq
   q_a} yields \(\delta_\rho(q,a) = \delta(q,a)\). Therefore \(\delta_\rho(q,a)= q\)
 and thus \(a \in E'\). Hence \(E \subseteq E'\).
        \end{itemize}
\end{proof}

The following lemma will be fruitful to show results on languages recognized by the automata \(\A(\rho)\).

\begin{lemma}\label{lemma:RAq subset RA'q}
	Let \(\A = \Auto\) be a non-linear minimal automaton recognizing an ideal, let \(\rho \in \setsep\) and let \(q \in \mifa{\rho}\). One has \(\Res{\A}{q} \subseteq \Res{\A(\rho)}{q}\).
\end{lemma}

\begin{proof}
Let \(q \in \mifa{\rho}\),
and let \(E = \set{a \in \Sigma}{\delta(q,a) = q}\) and \(E' = \set{a
  \in \Sigma}{\delta_\rho(q,a) = q}\). Lemma~\ref{lemma: E = E'} sets
that \(E = E'\). Let's show by a well founded induction on \(\preccurlyeq_{\A(\rho)}\)
that for all \(q \in \mifa{\rho}\) we have \(\Res{\A}{q} \subseteq
\Res{\A(\rho)}{q}\).

If \(\succsA{\A(\rho)}{q}=\{q\}\), then \(q=q_f\) and \(q\) is a sink state
in \(\A(\rho)\). It follows that \(E' = \Sigma\). Since \(q_f\) is also
the final state of \(\A\) and is also a sink state in \(\A\), one has
\(\Res{\A}{q} = \Sigma^* =\Res{\A(\rho)}{q}\).

Assume now that \(\succsA{\A(\rho)}{q} \neq \{q\}\). Since \(\A\) and \(\A(\rho)\) are  both deterministic one has
\[\Res{\A(\rho)}{q} = E'^*\bigcup\limits_{a \in \Sigma \setminus E'} a \Res{\A(\rho)}{\delta_\rho(q,a)} \quad \text{ and } \quad \Res{\A}{q} = E^*\bigcup\limits_{a \in \Sigma \setminus E} a \Res{\A}{\delta(q,a)} \enspace .\]

But \(E = E'\). Therefore it suffices to show that for all \(a \in
\Sigma \setminus E'\), \(\Res{\A}{\delta(q,a)}\) is included in \(\Res{\A(\rho)}{\delta_\rho(q,a)}\).
	
By definition of \(\delta_\rho\) there exists \(s \in \mifa{\rho}\) such that \(q
\preccurlyeq_\A s\) and  \(\delta(s,b) = \delta_\rho(q,b)\). Since \(q
\preccurlyeq_\A s\) there exists \(u \in \Sigma^*\) such that
\(\delta(q,u) = s\). Therefore \(\delta(q,ua) = \delta_\rho(q,a)\). Thus, using
Proposition~\ref{prop:new prop on ideals}.\ref{lemma:R(delta(a)) subset R(delta(ua))}, one has
\(\Res{\A}{\delta(q,a)} \subseteq \Res{\A}{\delta_\rho(q,a)}\). Also since
\(a \notin E'\) one has \(q \prec \delta_\rho(q,a)\). Hence, by induction
hypothesis \(\Res{\A}{\delta_\rho(q,a)} \subseteq
\Res{\A(\rho)}{\delta_\rho(q,a)}\). Finally since \(\Res{\A}{\delta(q,a)}
\subseteq \Res{\A}{\delta_\rho(q,a)}\) and \(\Res{\A}{\delta_\rho(q,a)}
\subseteq \Res{\A(\rho)}{\delta_\rho(q,a)}\), one can conclude that
\(\Res{\A}{\delta(q,a)} \subseteq \Res{\A(\rho)}{\delta_\rho(q,a)}\).
\end{proof}
\end{toappendix}

Proposition~\ref{lemma:RA'q' equality} applied to the initial state yields that \(\A(\rho)\) accepts an ideal. The proof is rather technical, relying both on the partially ordered structure of the automata and on the subword closure of ideals.

\begin{proprep}\label{lemma:RA'q' equality}
	Let \(\A = \minAuto\) be a non-linear minimal automaton that recognizes an ideal and \(\rho \in \setsep\). For all \(q \in \mifa{\rho}\), the language \(\Res{\A(\rho)}{q}\) is an ideal.
\end{proprep}

\begin{proof}
  
 \begin{figure}
	\centering
 	\includegraphics[width=0.7\textwidth]{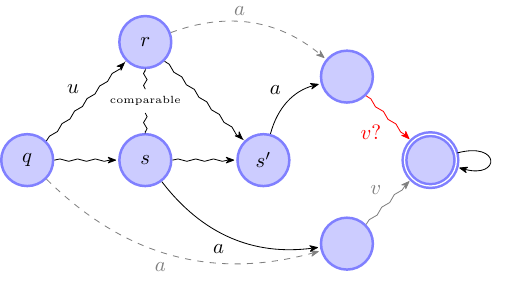}
	\caption{Illustration for proof of Proposition~\ref{lemma:RA'q' equality}. The grey transitions are the ones in \(\A(\rho)\) and the zigzags represent runs over words.}\label{fig:proof nonlin decomp}
 \end{figure}

  Let \(q\in \mifa{\rho}\) and \(E = \set{a \in \Sigma}{\delta_\rho(q,a) = q}\).

	Let's show by induction on \(\preccurlyeq_{\A(\rho)}\) that \(\Res{\A(\rho)}{q}\) is an ideal. In the base case, we do as in the previous proof to show that \(\Res{\A(\rho)}{q} = \Sigma^*\). Otherwise, we show that the structure of the automaton respects the schema described in Fig.~\ref{fig:proof nonlin decomp} and then we conclude by properties on ideals.

	\begin{itemize}
	\item Base case: the automaton \(\A(\rho)\) has a unique final state \(q_f\) which is the greatest element of \(\mifa{\rho}\) for \(\preccurlyeq_{\A(\rho)}\). If \(q=q_f\), \(q\) is a sink state and \(\Res{\A(\rho)}{q} = \Sigma^*\).

	\item If \(q \in \succsA{\A(\rho)}{\rho}\): Corollary~\ref{cor:A' when q leq q_a} ensures that \(\Res{\A(\rho)}{q} = \Res{\A}{q}\). And Proposition~\ref{prop: hold prop on ideals}.\ref{lemma:Res of ideal are ideal} provides that \(\Res{\A}{q}\) is an ideal. Therefore \(\Res{\A(\rho)}{q}\) is an ideal.

	\item If \(q\) is not a sink state in \(\A(\rho)\) and  \(q \notin \succsA{\A(\rho)}{\rho}\):
          by definition of a sink state \(E \neq \Sigma\). Set
			
	\[\mathcal{R} = \bigcup\limits_{a \in \Sigma \setminus E} \Shuffle{\left[a \Res{\A(\rho)}{\delta_\rho(q,a)}\right]}\enspace.\]

	The language \(\mathcal{R}\) is an ideal as it is a union of ideals. Thus it suffices to show that \(\Res{\A(\rho)}{q} = \mathcal{R}\) to conclude. Since \(\A(\rho)\) is deterministic one has, 

			\[\Res{\A(\rho)}{q} = E^* \cdot \bigcup\limits_{a \in \Sigma\setminus E} a \Res{\A(\rho)}{\delta_\rho(q,a)}\enspace.\]

			For all \(a\in \Sigma\setminus E\), since \(\varepsilon \in \Sigma^*\), by definition of the shuffle product, one has that \(a \Res{\A(\rho)}{\delta_\rho(q,a)}\) is a subset of \(\Shuffle{a \Res{\A(\rho)}{\delta_\rho(q,a)}}\). Therefore, since ideals are close by upper-words, one has \(\Res{\A(\rho)}{q} \subseteq \mathcal{R}\). Now let's show that \(\mathcal{R} \subseteq \Res{\A(\rho)}{q}\).

Let \(w \in \mathcal{R}\): there exists \(a \in \Sigma\setminus E\)
such that \(w \in \Shuffle{\left[a
    \Res{\A(\rho)}{\delta_\rho(q,a)}\right]}\).  Therefore,
there exist \(u \in \Sigma^*\) and \(v \in \Shuffle{\Res{\A(\rho)}{\delta_\rho(q,a)}}\) such that \(w =
uav\).
Since, by induction hypothesis, \(\Res{\A(\rho)}{\delta_\rho(q,a)}\) is an ideal,
\(w\in
\Res{\A(\rho)}{\delta_\rho(q,a)}\).  Let \(r = \delta_\rho(q,u)\). By
Lemma~\ref{lemma:existence of s_(q,b)}, the two following minima exists:
 \(s = \min\set{s'' \in \mifa{\rho}}{q \preccurlyeq_\A s''
  \wedge \delta(s'',a)\in \mifa{\rho}}\) and \(s' = \min\set{s'' \in
  \mifa{\rho}}{r \preccurlyeq_\A s'' \wedge \delta(s'',a)\in \mifa{\rho}}\).  By
definition of \(\delta_\rho\), one has \(\delta_\rho(q,a)= \delta(s,a)\) and
\(\delta_\rho(r,a)= \delta(s',a)\). Also, Lemma~\ref{lemma:A' when q leq
  q_a} implies that \(\delta_\rho(s,a) = \delta(s,a)\) and
\(\delta_\rho(s',a) = \delta(s',a)\).

We first show that \(s\) and \(r\) are comparable. Since \(q
\preccurlyeq_{\A(\rho)} \rho\) and since \(\delta(\rho,a) \in \mifa{\rho}\),
the minimality of \(s\) provides that \(s \preccurlyeq_{\A(\rho)} \rho\).
Hence, using Proposition~\ref{porp:rank and states}.\ref{cor:eq norm and preccurlyeq}, we have
\(\norm{s} \leq \norm{\sep}+1\).
If \(\norm{s} = \norm{\sep}+1\), since \(q \in \mifa{\rho}\)
Proposition~\ref{prop:sep}.\ref{cor:eq norm and preccurlyeq} ensures that \(q = \rho\).
Therefore since \(r \in \mifa{\rho}\), \(s\) and \(r\) are comparable. Otherwise if \(\norm{s} \leq \norm{\sep}\) Proposition~\ref{prop:sep}.\ref{lemma: above sep mifa = S} provides \(\mifa{s} = S\). And since \(r \in S\), one has that \(s\) and \(r\) are comparable.

Now we prove that \(s \preccurlyeq_{\A} s'\). Since \(s\) and \(r\)
are comparable, either \(r \preccurlyeq_{\A} s\) or \(s
\preccurlyeq_{\A} r\). If \(r \preccurlyeq_{\A} s\), by definition of
\(s\) and \(s'\) and unicity of minima, one has \(s = s'\), hence \(s \preccurlyeq_{\A}
s'\). Otherwise, if \(s \preccurlyeq_{\A} r\), since by definition \(r \preccurlyeq_\A s'\), the transitivity of \(\preccurlyeq\) gives that \(s \preccurlyeq_{\A} s'\). Hence, 
\(s' \preccurlyeq_{\A} s\) always holds.

Finally it remains to show that \(uav \in \Res{\A(\rho)}{q}\). We proceed  by a case distinction on whether \(av \in \Res{\A}{s}\) or not.

\begin{itemize}
      \item If \(av \in \Res{\A}{s}\):

       Since \(s \preccurlyeq_\A s'\) Proposition~\ref{prop:new prop on ideals}.\ref{lemma:preccurlyeq and
         inclusion of residuals} gives \(\Res{\A}{s} \subseteq
       \Res{\A}{s'}\). Hence, since \(av \in \Res{\A}{s}\), one has
       \(av \in \Res{\A}{s'}\). Then, applying Lemma~\ref{lemma:RAq
         subset RA'q} yields \(\Res{\A}{s'} \subseteq \Res{\A(\rho)}{s'}\)
       yields. Therefore \(av \in \Res{\A(\rho)}{s'}\). Thus, since
       \(\delta_\rho(s',a) = \delta_\rho(r,a)\) and since \(av \in
       \Res{\A(\rho)}{s'}\), one has \(av \in \Res{\A(\rho)}{r}\). Finally,
       since \(\delta_\rho(q,u) = r\), one can conclude \(uav \in
       \Res{\A(\rho)}{q}\).

	\item If \(av \notin \Res{\A}{s}\):

To prove \(w \in \Res{\A(\rho)}{q}\), we first show that \(\delta(s,a)
\prec_\A \rho\) and then we will prove that \(\delta(s,a)
\preccurlyeq_{\A(\rho)} \delta(s',a)\).

First let's show that \(\delta(s,a) \prec_\A \rho\) by
contradiction. Assume that \(\rho \preccurlyeq \delta(s,a)\). Since \(v
\in \Res{\A(\rho)}{\delta_\rho(s,a)}\), one has \(v \in
\Res{\A}{\delta(s,a)}\). Then, using Corollary~\ref{cor:A' when q leq
  q_a}, \(\Res{\A(\rho)}{\delta_\rho(s,a)} = \Res{\A}{\delta(s,a})\)
yields. Now, since \(\delta_\rho(s,a) = \delta(s,a)\), one has \(av \in
\Res{\A}{s}\), a contradiction.

Now let's show that \(\delta(s,a) \preccurlyeq_{\A}
\delta(s',a)\). Proposition~\ref{porp:rank and states}.\ref{cor:norm and path} applied on \(\delta(s,a)
\prec_{\A} \rho\) provides \(\norm{\delta(s,a)} < \norm{\sep}\).

Thus, by Lemma~\ref{lemma:unicity above sep} one has that \(\delta(s,a) \preccurlyeq \sep\). Hence applying Proposition~\ref{prop:sep}.\ref{lemma: above sep mifa = S} one can conclude that \(\delta(s',a) \in
\mifa{\delta(s,a)}\). We claim now that \(\delta(s,a)
\preccurlyeq_{\A} \delta(s',a)\). Since\(s \preccurlyeq_{\A} s'\),
there exists \(u' \in \Sigma^*\) such that \(\delta(s,u') =
s'\). Therefore \(\delta(s',a)~=~\delta(s,u'a)\). Thus, if
\(\delta(s',a) \preccurlyeq_{\A} \delta(s,a)\) then, since
\(\delta(s',a) = \delta(s,u'a)\), \(\delta(s,u'a)
\preccurlyeq_\A \delta(s,a)\) holds. Consequently, using Proposition~\ref{prop:new prop on ideals}.\ref{lemma:
  min automata preccurlyeq equality}, one has \(\delta(s,a) =
\delta(s,u'a) = \delta(s',a)\). Thus, \(\delta(s,a) \preccurlyeq_{\A}
\delta(s',a)\). Otherwise, if  \(\delta(s',a) \not\preccurlyeq_{\A}
\delta(s,a)\), since \(\delta(s',a) \in \mifa{\delta(s,a)}\), we would
have \(\delta(s,a) \preccurlyeq_{\A} \delta(s',a)\), proving the
claim.

Consequently, by Lemma~\ref{lemma: preccurlyeq in A' and A}, one has \(\delta(s,a) \preccurlyeq_{\A(\rho)} \delta(s',a)\).

Finally let's show that \(w = uav \in \Res{\A(\rho)}{q}\). Since \(a
\notin E\) and by definition of \(\preccurlyeq\) one has \(q \prec \delta(q,a)\). Hence by construction of \(s\) one has \(q \prec_{\A(\rho)} \delta(s,a)\). Thus, by induction
hypothesis, \(\Res{\A(\rho)}{\delta(s,a)}\) is an ideal. Hence since,
\(\delta(s,a) \preccurlyeq_{\A(\rho)} \delta(s',a)\), by
Proposition~\ref{prop:new prop on ideals}.\ref{lemma:preccurlyeq and inclusion of residuals},
\(\Res{\A(\rho)}{s} \subseteq \Res{\A(\rho)}{s'}\) yields . Since \(av \in
\Res{\A(\rho)}{q}\) and since \(\A(\rho)\) is deterministic, one has \(v \in
\Res{\A(\rho)}{\delta_\rho(q,a)}\). But \(\delta_\rho(q,a) = \delta(s',a)\),
therefore one has \(v \in \Res{\A(\rho)}{\delta(s',a)}\). Finally, since
\(\Res{\A(\rho)}{s} \subseteq \Res{\A(\rho)}{s'}\), one can conclude that
\(v \in \Res{\A(\rho)}{\delta(s',a)}\). But since \(\delta(s',a) =
\delta_\rho(r,a)\), one has \(v \in \Res{\A(\rho)}{\delta_\rho(r,a)}\). Thus by
definition of a residual, \(av \in \Res{\A(\rho)}{r}\), and since \(r =~\delta_\rho(q,u)\), one has \(uav \in \Res{\A(\rho)}{q}\).
\end{itemize}

Finally \(uav \in \Res{\A(\rho)}{q}\) hence since, \(w =uav\), one has
\(w \in \Res{\A(\rho)}{q}\). Therefore \(\mathcal{R} \supseteq \Res{\A(\rho)}{q}\).
	\end{itemize}

	Finally we proved that \(\mathcal{R} \subseteq
        \Res{\A(\rho)}{q}\). Combined with \(\Res{\A(\rho)}{q} \subseteq
        \mathcal{R}\), it  gives that \(\Res{\A(\rho)}{q} = \mathcal{R}\).
\end{proof}

Finally, to prove Theorem~\ref{theo:non-linear are decomposable}, it
remains to show that the language of \(\A\) is
equal to the intersection of the languages of its family automata.

\begin{proprep}\label{lemma:L(A) = inter L(A(a))}
	Let \(\A = \minAuto\) be a non-linear minimal automaton that recognizes an ideal. One has \[\lang{\A} = \bigcap\limits_{\rho \in \setsep} \lang{\A(\rho)} \enspace.\]
\end{proprep}

The proof is deferred to the appendix; the key idea is that,by construction, \(\lang{\A} \subseteq \lang{\A(\rho)}\), one
directly has that every word accepted by \(\A\) is also accepted by
all the \(\A(\rho)\). The reverse inclusion can be showed by proving that
for any \(w\notin \lang{\A}\), there exists a \(\A(\rho)\) that doesn't
  accept it.

\begin{proof}
For all \(\rho \in \setsep\), Lemma~\ref{lemma:RAq subset RA'q} ensures
that \(\Res{\A}{\iota} \subseteq \Res{\A(\rho)}{\iota}\). Hence
\(\lang{\A} \subseteq \lang{\A(\rho)}\). Thus one has \[\lang{\A}
\subseteq \bigcap\limits_{\rho \in \setsep} \lang{\A(\rho)} \enspace.\]

Now let us show the reverse inclusion. Let \(w = w_1\dots w_n \notin
\lang{\A}\) and let \(q_1, \dots, q_{n+1} \in S\) verifying \(q_1 =
\iota\) and for all \(i \in \setint{n}, q_{i+1} = \delta(\iota,
w_1\dots w_i)\). One has for all \(i \in \setint{n}, q_{i}
\preccurlyeq_\A q_{i+1}\).

By Proposition~\ref{prop:sep}.\ref{lemma: above sep mifa = S}, \(\sep\) is comparable to
every state in \(S\).  Therefore, since \(\A\) is trim, \(\iota
\preccurlyeq \sep\). Since \(q_1 = \iota\), there exists \(m =
\max\set{i \in\setint{n+1}}{q_i \preccurlyeq \sep}\).

Now we claim that there exists \(\rho \in \setsep\) such that for all
\(i \in \setint{n}\), one has \(q_i \in \mifa{\rho}\).

If \(m = n+1\), then let \(\rho\) be any state of \(\setsep\). By definition of \(m\), for all \(i \in
\setint{n+1}, q_i \preccurlyeq \sep\). Hence, since \(\sep
\preccurlyeq_\A \rho\), one has \(q_i \preccurlyeq \rho\).

Now if \(m < n+1\), we define \(\rho\) as follows. Since \((q_i)_{i \in \setint{n+1}}\) is an increasing sequence,
by definition of maximum, \(\sep \prec q_{m+1}\). Therefore
\(\norm{q_{m+1}} \geq \norm{\sep} + 1\) is given by
Proposition~\ref{porp:rank and states}.\ref{cor:norm and path}. Thus, by Proposition~\ref{porp:rank and states}.\ref{lemma:lower than
  n implis path taht uses n}, there exists \(\rho \in \setsep\) such that
\(\rho \preccurlyeq q_{m+1}\). Consequently, since
\((q_i)_{i \in \setint{n}}\) is an increasing sequence and by
transitivity of \(\preccurlyeq\), for all \(i \in \llbracket m+1, n+1
\rrbracket, \rho \preccurlyeq q_i\). Similarly and since
\(\sep \preccurlyeq \rho\) for all \(i \in \setint{m}, q_i
\preccurlyeq \rho\). Hence by definition of \(\mifa{\rho}\), for all \(i \in
\setint{n+1}, q_i \in \mifa{\rho}\).

Recall that \(\A(\rho) = (\mifa{\rho}, \Sigma, \iota, F, \delta_\rho)\). For
any \(i \in \setint{n+1}\), since \(q_{i+1} \in \mifa{\rho}\),
Lemma~\ref{lemma:A' when q leq q_a} provides that
\((q_i,w_i,q_{i+1})\in \delta_\rho\). Thus one has \(\delta_\rho(\iota,w) =
q_n\).  Since  \(w \notin \lang{\A})\), it follows that \(q_n \notin
F\). Therefore, since \(q_n \notin F\) and \(\delta_\rho(\iota,w) = q_n\),
one has \(w \notin \lang{\A(\rho)}\). Consequently \(w \notin
\bigcap\limits_{\rho \in \setsep} \lang{\A(\rho)}\) proving
that \[\bigcap\limits_{\rho \in \setsep} \lang{\A(\rho)} \subseteq
\lang{\A} \enspace,\] which concludes the proof.
\end{proof}

\section{Intersection-decomposition of linear automata}\label{sec:linear_automata_and_decomposition}

\begin{figure}
	\centering
	\begin{subfigure}{0.39\textwidth}
		\centering
		\includegraphics[width=1\textwidth]{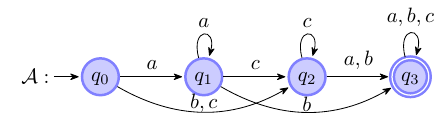}
	\end{subfigure}
	\hfill
	\begin{subfigure}{0.29\textwidth}
		\centering
		\includegraphics[width=1\textwidth]{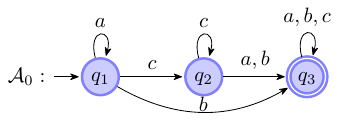}
	\end{subfigure}
	\hfill
	\begin{subfigure}{0.29\textwidth}
		\centering
		\includegraphics[width=1\textwidth]{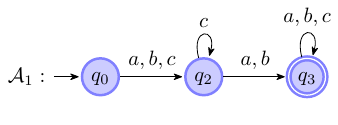}
	\end{subfigure}
	\caption{An automaton \(\A\),
		along with two smaller automata
		satisfying \(\lang{\A_0} \cap \lang{\A_1}\ = \lang{\A}\).}
	\label{fig: decomposition}
\end{figure}

The goal of this section is to prove Proposition~\ref{prop: decomposition of linear automata},
which characterizes intersection-composite linear automata recognizing ideals.
Therefore, during the whole section,
we use terms such as ``composite'' and ``prime''
to specifically refer to the \emph{intersection} decomposition.

We introduce notations adapted to linear automata.
For any linear automaton \(\A\),
we denote it set of states as \(\{q_0,q_1,\ldots,q_n\}\),
where the state enumeration agrees with the transition relation:
\(q_i \preccurlyeq q_j \text{ if and only if } i \leq j\).
Remark that \(\iota = q_0\)
as otherwise \(q_0\) is not reachable.
Moreover, if \(\A\) is a minimal automaton recognizing an ideal, it set of final states is \(\{q_n\}\)
as Proposition~\ref{prop: hold prop on ideals}.\ref{theo:ideal are regular} implies that such automata
have one final sink state, and 
\(q_n\) is the only sink state of \(\A\).

\begin{figure}[h!]
 	\centering
  	\includegraphics[width=7cm]{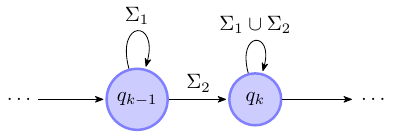}
	\caption{Illustration of a damping pattern between \(q_{k-1}\) and \(q_{k}\).}
	\label{fig:damping pattern}
\end{figure}

\noindent
\textbf{\textsf{Damping pattern.}}
We characterize composite linear automata through the notion of \emph{damping patterns}.
A damping pattern of a linear automaton is a pair of consecutive states
such that each letter that either keeps the first state unchanged
or moves from the first state to the second also induces a loop on the second state.
Formally, let \(\A = \AutoLinGen\) be a linear automaton.
For all \(i,j \in \setint{n}\),
we denote by \(\Sigma_{i,j}(\A)\), or simply \(\Sigma_{i,j}\) when the automaton is clear,
the set \(\set{a \in \Sigma}{\delta(q_i,a)= q_{j}}\)
of letters labelling the transitions from \(q_i\) to \(q_j\).
For all \(k \in \setint{n-1}\),
\(\A\) has a damping pattern
between \(q_{k-1}\) and \(q_{k}\)
if
\[
\partAlph{k-1}{k-1} \cup \partAlph{k-1}{k}
\subseteq
\partAlph{k}{k}.
\]
For instance, this pattern is illustrated in Fig.\ref{fig:damping pattern}. And the automaton \(\A\) in Fig.~\ref{fig: decomposition}
has a damping pattern between \(q_0\) and \(q_1\). Indeed in that case \(\Sigma_{0,0} \cup \Sigma_{0,1} = \{a\} = \Sigma_{1,1}\).
The figure shows that
\(\A\) is composite, which is consistent with the presence of a damping pattern:
the following proposition states that damping patterns characterize composite linear automata.

\begin{restatable}[]{theorem}{decompOfLinear}
	\label{prop: decomposition of linear automata}
	Let \(\A\) be a linear minimal automaton that recognizes an ideal.
	Then \(\A\) is prime if and only if it has no damping pattern.
	Moreover, if \(\A\) is composite it is decomposable into two linear automata recognizing ideals.
\end{restatable}
The proof of Theorem~\ref{prop: decomposition of linear automata} in the case where \(\A\) linear is done in the next two subsections.

In Section~\ref{subsec:DecomposableLinearAutomata},
we show how to decompose automata that have a damping pattern
between two states \(q_{k-1}\) and \(q_{k}\).
To this end, we define the reduced automata \(\A_{k-1}\) and \(\A_{k}\),
obtained by removing the state \(q_{k-1}\) (resp. \(q_{k}\))
from \(\A\), and redirecting the transitions leading towards the removed state
to its successor \(q_{k}\) (resp. \(q_{k+1}\)).
We prove that both automata recognize ideals (Proposition~\ref{lemma:A_k recognize an ideal language}),
and that they satisfy \(\lang{\A} = \lang{\A_{k-1}} \cap \lang{\A_{k}}\) (Proposition~\ref{lemma:L(A) = inter L(A_k)}).
Since both automata have strictly fewer states than \(\A\), this establishes that  \(\A\) is composite.

In Section~\ref{subsec:PrimeLinearAutomata}
we show that the automata without damping patterns are prime.
To that end, we rely on the notion of \emph{primality witness} of an automaton \(\A\), which is a word \(w \in \Sigma^*\) that does not belong to \(\lang{\A}\)
but is in the language recognized by every automaton
that is strictly smaller than \(\A\) and recognizes a superset of \(\lang{\A}\).  
Formally, we construct a set of words \(\Wit(\A) \subseteq \Sigma^*\)
that is non-empty as long as \(\A\) does not have a damping pattern (Lemma~\ref{lemma:cond is eq Wit non-empty}).
Then, we prove that \(\Wit(\A) \cap \lang{\A} = \emptyset\) (Proposition~\ref{lemma:wit notin L(A)})
yet \(\Wit \subseteq \lang{\A'}\) for every automaton \(\A'\) satisfying \(|\A'| < |\A|\) and \(\lang{\A} \subseteq \lang{\A'}\) (Proposition~\ref{lemma:wit in L(A')}).
Therefore the intersection of every sequence of languages recognized
by automata strictly smaller than \(\A\)
either does not contain \(\lang{\A}\) or contains \(\Wit\),
which proves that \(\A\) is prime.

\subsection{Decomposition of automata with a damping pattern}\label{subsec:DecomposableLinearAutomata}

In this section,
we show that every automaton with a damping pattern can be decomposed into two automata recognizing ideals.
We begin by defining, for every \(k \in \setint{n-1}\),
an automaton \(\A_k\), called the reduced automaton,
obtained by removing the state \(q_k\) from \(\A\)
and redirecting the transitions heading to \(q_k\) toward \(q_{k+1}\).
We then show that, if \(\A\) recognizes an ideal, so does each reduced automaton.
Finally, we prove that every pair of reduced automata
corresponding to a damping pattern
yields a decomposition of \(\A\).

\medskip
\noindent
\textbf{\textsf{Reduced automaton.}}
Let \(\A = \AutoLin\) be a linear automaton.
For every \(k \in \setintz{n-1}\),
the \emph{reduced automaton} \(\A_k\) is defined as follows.
Let
\[\delta_k^- = \set{(q,a,q_{k}) \in \delta}{q \in S} \text{ and } \delta_k^+ = \set{(q,a,q_{k+1})}{(q,a,q_k) \in \delta}.\]
Then we set \(\A_k = (\{q_0,q_1,\ldots,q_n\} \setminus\{q_k\}, \Sigma, \iota_k, \{q_n\}, \delta_k)\),
where
\[
\delta_k = (\delta \setminus \delta_k^-)\cup \delta_k^+ \quad
\textup{ and }  \quad
\iota_k =
\left\{
\begin{array}{ll}
	q_1 &
	\textup{if \(k = 0\);}\\
	q_0 &
	\textup{if \(k > 0\)}.
\end{array}
\right.\]
Fig.~\ref{fig: decomposition} illustrates this construction by showing an automaton \(\A\) along 
with the automata \(\A_0\) and \(\A_1\) obtained by this process.
Remark that this construction guarantees that each reduced automaton is linear.
Our first result 
states that if \(\A\) recognizes an ideal,
so do its reduced automata.

\begin{proprep}\label{lemma:A_k recognize an ideal language}
	Let \(\A\) be a linear automaton
	that recognizes an ideal.
	For all \(k \in \setintz{n-1}\),
	the automaton \(\A_k\) also recognizes an ideal.
\end{proprep}

\begin{proof}
	Let \(k \in \setintz{n-1}\),
	and let \(\A_k =
	(\{q_0,q_1,\ldots,q_n\} \setminus\{q_k\}, \Sigma, q_0, \{q_n\}, \delta_k)\).
	We show by induction on \(n-i\)
	that \(\Res{\A_k}{q_i}\) is an ideal
	for all \(i \in \{0,1,\ldots,n\} \setminus \{k\}\).
	
	If \(i = n\), the fact that \(q_n\) is the accepting sink state
	of both \(\A\) and \(\A_k\) yields
	\(\Res{\A_k}{q_i} = \Sigma^*\), which is an ideal.
	Now suppose that \(i < n\).
	We prove that \(\Res{\A_k}{q_i} = \Res{\A_k}{q_i} \shuffle \Sigma^*\).
	Let \(w \in \Res{\A_k}{q_i} \shuffle \Sigma^*\).
	Then there exist \(u \in \Res{\A_k}{q_i}\) and \(v \in \Sigma^*\)
	satisfying \(w \in u \shuffle v\).
	Note that, as \(u \in  \Res{\A_k}{q_i}\) and \(q_i\) is not an accepting state,
	the run of \(\A_k\) over \(u\) starting from \(q_i\) leaves \(q_i\) at some point.
	Let \(u = u_1au_2\) denote the decomposition of \(u\)
	such that the run of \(\A_k\) over \(u\) starting in \(q_i\)
	remains in \(q_i\) while reading \(u_1\)
	and leaves \(q_i\) upon reading the letter \(a\).
	Moreover, consider the decomposition \(w = w_1aw_2\) of \(w\)
	where \(w_1a\) is the shortest prefix of \(w\)
	such that \(u_1a\) is a subword of \(w_1a\).
	Note that \(u_2\) is a subword of \(w_2\).
	
	We study the runs of \(\A_k\)
	over \(u\) and \(w\) starting in \(q_i\).
	First, the definition of \(u_1\) yields
	\begin{equation}\label{eq:loop}
		\delta_k(q_i,u_1a) = \delta_k(q_i,a).
	\end{equation}
	Then, Proposition~\ref{prop:new prop on ideals}.\ref{lemma: min automata preccurlyeq equality}
	yields
	 \(\delta(q_i,a) \preccurlyeq_\A \delta(q_i,w_1a)\).
	Since \(\delta(q_i,w_1a) \preccurlyeq_\A \delta_k(q_i,w_1a)\)
	 by definition of \(\delta_k\), 
	 and since \(\delta_k(q_i,a)\) is either equal to \(\delta(q_i,a)\)
	 or to its successor state in the specific case where
	 \(\delta(q_i,a) = q_k\),
	 we get
	\begin{equation}\label{eq:jump}
		\delta_k(q_i,a) \preccurlyeq_\A \delta_k(q_i,w_1a).
	\end{equation}
	Finally, the fact that \(u \in \Res{\A_k}{q_i}\)
	implies that \(u_2 \in \Res{\A_k}{\delta_k(q_i,u_1a)}\).
	Therefore, the induction hypothesis allows the use of
	Proposition~\ref{prop:new prop on ideals}.\ref{lemma:preccurlyeq and inclusion of residuals},
	which yields \(u_2 \in \Res{\A_k}{\delta_k(q_i,w_1a)}\).
	Then, as \(u_2\) is a subword of \(w_2\),
	a second application of the induction hypothesis yields 
	\(w_2 \in \Res{\A_k}{\delta_k(q_i,w_1a)}\).
	Combined with Equations~\eqref{eq:jump} and \eqref{eq:loop},
	this proves that \(w \in~\Res{\A_k}{q_i}\).
\end{proof}

\noindent
To conclude,
we show that the reduced automata can be used to form a decomposition of \(\A\).

\begin{proprep}\label{lemma:L(A) = inter L(A_k)}
	Let \(\A\) be a linear automaton
	that recognizes an ideal.
	If \(\A\) has a damping pattern between two states \(q_{k-1}\) and \(q_k\),
	then \(\lang{\A} = \lang{\A_{k-1}} \cap \lang{\A_{k}}\).
\end{proprep}

\begin{proof}
	Suppose that \(\A = \AutoLin\) has a damping pattern between two states \(q_{k-1}\) and \(q_k\).
	To prove that \(\lang{\A} = \lang{\A_{k-1}} \cap \lang{\A_{k}}\),
	we first show that \(\lang{\A}\) is contained in the language of \emph{every} reduced automaton.
	We then show that the damping pattern implies that \(\lang{\A_{k-1}} \cap \lang{\A_{k}} \subseteq \lang{\A}\).

	\subparagraph*{Containment in reduced automata.}
	Let \(i \in \setintz{n-1}\).
	For all \(u \in \Sigma^*\),
	we show by induction over \(|u|\)
	that \(\delta(q_0,u) \preccurlyeq_\A \delta_i(\iota_i,u)\).
	
	If \(u = \varepsilon\),
	we have \(\delta(q_0,u) = q_0 \preccurlyeq_\A \iota_i = \delta_i(\iota_i,u)\).
	Otherwise, let \(u = va\) 
	with \(v \in \Sigma^*\) and
	\(a \in \Sigma\).
	Then the induction hypothesis yields
	\(\delta(q_0,v) \preccurlyeq_\A \delta_i(\iota_i,v)\),
	that is, there exists \(w \in \Sigma^*\)
	satisfying \(\delta(\delta(q_0,v),w) = \delta_i(\iota_i,v)\).
	As required, we have
	\[
		\delta(q_0,u)
		=
		\delta(\delta(q_0,v),a)
		\stackrel{(1)}{\preccurlyeq_\A}
		\delta(\delta(q_0,v),wa)
		=
		\delta(\delta_i(\iota_i,v),a)
		\stackrel{(2)}{\preccurlyeq_\A}
		\delta_i(\delta_i(\iota_i,v),a) =
		\delta_i(\iota_i,u).
	\]
	Inequality (1) follows
	from Proposition~\ref{prop:new prop on ideals}.\ref{lemma: min automata preccurlyeq equality}
	since the states \(\delta(\delta(q_0,v),a)\) and \(\delta(\delta(q_0,v),wa)\)
	are comparable as \(\A\) is linear.
	Inequality (2) follows directly from the definition 
	of \(\delta_i\).
	
	Finally, for all \(u \in \lang{\A}\)
	the property above yields 
	\(q_n = \delta(q_0,u)\preccurlyeq_{\A} \delta_i(\iota_i,u)\).
	This implies \(u \in \lang{\A_i}\) as \(q_n\) is the unique maximal accepting state of \(\A\) and \(\A_i\). Hence \(\lang{\A} \subseteq \lang{\A_i}\).

	\subparagraph*{Reverse inclusion.}
	Now we prove that \(\lang{\A_{k-1}} \cap \lang{\A_{k}} \subseteq \lang{\A}\).
	We show that every word \emph{rejected}
	by \(\A\) is not in \(\lang{\A_{k-1}} \cap \lang{\A_{k}}\).
	Let \(u \notin \lang{\A}\),
	and let us consider the run \(\pi\) of \(\A\) over \(u\).
	We distinguish cases depending on whether or not \(q_{k-1}\)
	and \(q_k\) occur along \(\pi\),
	and show that in each case \(u \notin \lang{\A_{k-1}} \cap \lang{\A_{k}}\).
	
	If the state \(q_k\) never occurs along \(\pi\),
	then the run of \(\A_k\) over \(u\)
	coincides with that of \(\A\), hence \(u \notin \lang{\A_{k}}\), and therefore
	\(u \notin \lang{\A_{k-1}} \cap \lang{\A_{k}}\).
	Similarly, if \(q_{k-1}\) never occurs along \(\pi\)
	then \(u \notin \lang{\A_{k-1}}\) and again \(u \notin \lang{\A_{k-1}} \cap \lang{\A_{k}}\).
	
	It remains to consider the case where both \(q_{k-1}\)
	and \(q_k\) occur along \(\pi\).
	Let \(u = u_1 a_1 u_2 a_2 u_3\) denote the decomposition of \(u\)
	where \(u_1,u_2,u_3 \in \Sigma^*\)
	and \(a_1,a_2 \in \Sigma\) such that
	\begin{itemize}
		\item \(a_1\) is the letter upon which \(\pi\) first reaches \(q_{k-1}\);
		\item \(a_2\) is the letter upon which \(\pi\) first reaches \(q_{k}\).
	\end{itemize}
	Note that, as \(\A\) is linear,
	the run \(\pi\) remains in \(q_{k-1}\) while processing \(u_2\), 
	that is, \(u_2 \in \partAlph{k-1}{k-1}^*\).
	
	Consider the run \(\pi'\) of \(\A_{k-1}\) over \(u\).
	The prefix of \(\pi'\) over \(u_1\)
	coincides with that of \(\pi\) as \(q_{k-1}\) is not visited along this prefix.
	The two runs diverge upon reading \(a_1\): the run \(\pi\) transitions to \(q_{k-1}\),
	while \(\pi'\) follows the redirected transition to \(q_k\).
	Since \(\A\) has a damping pattern between \(q_{k-1}\) and \(q_k\),
	the run \(\pi'\) remains in \(q_k\) while \(\pi\) processes \(u_2\)
	(since \(\partAlph{k-1}{k-1} \subseteq \partAlph{k}{k}\))
	and also while \(\pi\) moves from \(q_{k-1}\) to \(q_k\) upon reading \(a_2\)
	(since \(\partAlph{k-1}{k} \subseteq \partAlph{k}{k}\)).
	Thus, after reading \(u_1 a_1 u_2 a_2\), both runs are in the same state \(q_k\).
	Moreover, since \(q_{k-1}\) does not occur in the suffix \(u_3\),
	the runs \(\pi\) and \(\pi'\) coincide again on this suffix,
	hence they end in the same state.
	Since \(u \notin \lang{\A}\), we have that \(u \notin \lang{\A_{k-1}}\), and in turn
	\(u \notin \lang{\A_{k-1}} \cap \lang{\A_{k}}\).
\end{proof}

\subsection{Primality witnesses of automata without damping patterns}\label{subsec:PrimeLinearAutomata}

For every linear automaton \(\A = \Auto\)
we define a set of words \(\Wit(\A)\),
we show that this set is not empty if \(\A\) does not have a damping pattern (Lemma~\ref{lemma:cond is eq Wit non-empty}),
and that all the words in this set are primality witnesses of \(\A\) (Lemma~\ref{lemma:wit in L(A')}).

\subparagraph*{Witness set.} Let \(\A\) be a linear automaton without damping pattern.
The witnesses of \(\A\) are defined by concatenating smaller words as follows:
\[
\Wit(\A)  = \set{w(1) \ldots w(n-1) \in \Sigma^*}{\delta(q_{i-1},w(i)) = q_{i} \neq \delta(q_{i},w(i)) \textup{ for every } i \in \setint{n-1}}.
\]
We denote the elements of \(\Wit(\A)\) by their decomposition in words \(w(1) \ldots w(n-1)\),
and implicitly assume that this decomposition satisfies the above condition.

\begin{lemma}\label{lemma:cond is eq Wit non-empty}
	Every linear automaton
	\(\A\) without damping pattern satisfies
	\(\Wit(\A) \neq \emptyset\).
\end{lemma}

\begin{proof}
	Let \(\A = \AutoLinGen\) be a linear automaton without damping pattern.
	Then for every \(i \in \setint{n-1}\) we have
	\((\partAlph{i-1}{i-1} \cup \partAlph{i-1}{i})  \setminus \partAlph{i}{i} \neq \emptyset\),
	and we can set  
	\[
	w(i) = \left\{
	\begin{array}{ll}
		a \in \partAlph{i-1}{i} \setminus \partAlph{i}{i} 
		& \text{ if }\partAlph{i-1}{i} \setminus \partAlph{i}{i} \neq \emptyset;\\
		a_1a_2 \in (\partAlph{i-1}{i-1} \setminus \partAlph{i}{i})\cdot \partAlph{i-1}{i} 
		& \text{ if }\partAlph{i-1}{i} \setminus \partAlph{i}{i} = \emptyset
		\text{ and }\partAlph{i-1}{i-1} \setminus \partAlph{i}{i} \neq \emptyset.
	\end{array}
	\right.
	\]
	Both cases imply that \(\delta(q_{i-1},w(i)) = q_{i} \neq \delta(q_{i},w(i))\)
	as required by the definition of \(\Wit(\A)\),
	hence the word \(w(1)w(2) \ldots w(n-1)\) is in \(\Wit(\A)\).
\end{proof}
To prove that each word \(\wit = w(1) \ldots w(n-1) \in \Wit(\A)\) is a primality witness, 
we rely on the following technical lemma stating that,
while \(\wit \notin \lang{\A}\),
duplicating any \(w(i)\) yields a word in \(\lang{\A}\).

\begin{proprep}\label{lemma:wit notin L(A)}
	Let \(\A\) be a linear minimal automaton that recognizes an ideal.
	Every \(\wit = w(1) \dots w(n-1) \in \Wit(\A)\)
	satisfies \(\wit \notin \lang{\A}\)
	and
	\[w(1) w(2)\dots w(i-1) w(i) w(i) w(i+1) \dots w(n-1) \in \lang{\A}
	\textup{ for all }
	i \in \setint{n-1}.\]
\end{proprep}

\begin{proof}
	The definition of \(\Wit(\A)\) yields that
	\(\delta(q_{i-1},w(i)) = q_{i}\) for all \(i \in \setint{n-1}\).
	Stitching these runs together yields that
	\begin{equation}\label{equ:witnessStandardRun}
		\delta(q_{0},w(1)w(2) \ldots w(i)) = q_{i} \textup{ for all } i \in \setint{n-1}.
	\end{equation}
	In particular \(\delta(q_0,\wit) = q_{n-1}\),
	and since  \(q_n\) is the only accepting state of \(\A\) by Proposition~\ref{prop: hold prop on ideals}.\ref{lemma: ideals unique final state},
	we get that \(\wit\) is not in  \(\lang{\A}\).
	
	We now prove that duplicating any factor \(w(i)\) yields an accepted word.
	The definition of \(\Wit(\A)\) says that for all \(i \in\setint{n-1}\), 
	\(\delta(q_{i},w(i)) \neq q_{i}\).
	Since the transitions of \(\A\) agree with the state ordering, 
	this implies that \(\delta(q_{i},w(i)) \succcurlyeq q_{i+1}\),
	and, more generally,  \(\delta(q,w(i)) \succcurlyeq q_{i+1}\) for every \(q \succcurlyeq q_i\).
	Applying this observation \(j-i+1\) times yields that
	\begin{equation}\label{equ:witnessAdvancedRun}
		\delta(q,w(i)w(i+1) \ldots w(j)) \succcurlyeq q_{j+1} \textup{ for all }  i \leq j \in \setint{n-1} \textup{ and } q \succcurlyeq q_i.
	\end{equation}
	As a consequence, for every \(i \in \setint{n-1}\),
	\[\begin{array}{lll}
    	\multicolumn{3}{c}{\delta(q_0,w(1) w(2)\dots w(i-1) w(i) w(i) w(i+1) \dots w(n-1))}\\
    	 & = & \delta(q_i,w(i) w(i+1) \dots w(n-1))\\
    	& \succcurlyeq & q_n,
	\end{array}\]

	where the first equality follows from Equation~\eqref{equ:witnessStandardRun}
	and the second from Equation~\eqref{equ:witnessAdvancedRun}.
	Since \(q_n \in F\) is an accepting sink state, the result follows.
\end{proof}

\begin{proprep}\label{lemma:wit in L(A')}
	Let \(\A\) be a linear minimal automaton that recognizes a non-empty ideal.
	For every automaton \(\A'\) smaller than \(\A\),
	\(\lang{\A} \subseteq \lang{\A'}\) implies \(\Wit(\A) \subseteq \lang{\A'}\).
\end{proprep}

\begin{proof}
	Let \(\A = \AutoLin\) be a linear minimal automaton that recognizes a non-empty ideal,
	and let \(\A' = (S',\Sigma, \iota', F', \delta')\) be an automaton satisfying \(|S'| < n+1\) and
	\(\lang{\A} \subseteq \lang{\A'}\).
	We use Proposition~\ref{lemma:wit notin L(A)} and the fact
	that ideals are closed by upper-word
	to show that \(\A'\) reaches an accepting state while reading each word of \(\Wit(\A)\).
	
	Formally, let \(\wit = w(1) \dots w(n-1) \in \Wit(\A)\),
	and let \((r_i)_{i \in \llbracket 0, n-1 \rrbracket}\) be the sequence of states
	visited while reading each consequent word that makes up \(\wit\), that is,
	\(r_0 = \iota'\) and \(\delta'(r_{i-1},w(i))= r_i\) for all \(i \in \setint{n-1}\).
	We distinguish two cases depending on whether or not a state is repeated in this sequence,
	and show that \(\wit \in \lang{\A'}\) in both cases.
	
	\textbf{Case 1.} Suppose that there \(i,j \in \llbracket 1, n-1 \rrbracket\) with \(i < j\) such that \(r_i = r_j\).
	Consider the following words obtained by iterating factors of \(\wit\):
	\[
	\begin{array}{lll}
		\wit' & = & w(1) w(2) \dots w(i-1) w(i) w(i) w(i+1) \dots w(n-1);\\
		\wit'' & = & w(1) w(2) \dots w(i-1) (w(i) \dots w_{j-1})^2 w(j) \dots  w(n-1)).
	\end{array}
	\]
	Proposition~\ref{lemma:wit notin L(A)} implies that \(\wit' \in \lang{\A} \subseteq \lang{\A'}\).
	As a consequence, \(\wit'' \in \lang{\A} \subseteq \lang{\A'}\)
	as \(\wit'\) is a subword of \(\wit''\) and \(L(\A)\) is an ideal.
	Finally, the fact that \(r_i = r_j\)
	implies that \(\delta'(\iota,\wit) = \delta'(\iota',\wit'')\), 
	thus in turn \(\wit \in \lang{\A'}\).
	
	\textbf{Case 2.} Suppose that \(r_i \neq r_j\) for all \(i \neq j \in \llbracket 0, n-1 \rrbracket\).
	Since \(|S'| < n + 1\), this implies that all states of \(\A'\) are visited while reading \(\wit\).
	As \(\lang{\A}\) is non-empty, there exists \(v \in \lang{\A}\).
	Consequently, there exists a decomposition \(\wit = xy\) such that \(\delta'(\iota',v) = \delta'(\iota',x)\).
	It follows that \(\delta'(\iota',vy) = \delta'(\iota',xy) = r_{n-1}\).
	Moreover, \(vy \in \lang{\A} \subseteq \lang{\A'}\) as \(\lang{\A}\) is an ideal and \(v\) is a sub-word of \(vy\).
	Hence \(r_{n-1}\) is an accepting state, thus  \(\wit \in \lang{\A'}\). 
\end{proof}

\section{Union decomposition}\label{sec:union_decomposition}

Throughout this section, the notions of decomposition, composite and prime
refer to union decomposition.
The main result is  Theorem~\ref{theo:union characterization},
which characterizes prime minimal automata recognizing ideals.
For composite automata,
we show how to construct a decomposition into  \(|\Lmin|\) prime
automata recognizing ideals.
The characterization relies on the
notion of \emph{accelerating patterns}.

\begin{figure}[h!]
  \centering
    \includegraphics[width=7cm]{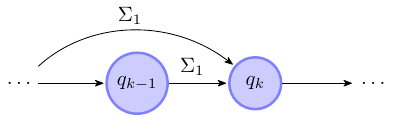}
  \caption{Illustration of an accelerating pattern between \(q_{k-1}\) and \(q_{k}\).}
  \label{fig:accelerating pattern}
\end{figure}

\subparagraph*{Accelerating pattern.} An \emph{accelerating pattern} in
a linear automaton recognizing an ideal is a state such that each letter
labelling an incoming transition from its greatest predecessor also
labels an incoming transition from a strictly smaller state.
Formally, let \(\A = \AutoLin\) be a linear automaton recognizing an ideal
such that \(q_{k-1}\prec q_{k}\) for all \(k \in \llbracket 1, n\rrbracket\).
For all \(k \in \llbracket 1, n\rrbracket\), \(\A\) has an accelerating pattern if there exists \(k\) such that
\[
\partAlph{k-1}{k}
\quad \subseteq \quad
\bigcup\limits_{j=0}^{k-2}
\partAlph{j}{k}.
\]
For instance,this pattern is illustrated in Fig.\ref{fig:accelerating pattern}. And the automaton \(\A\) in
Fig.\ref{fig: decomposition} has an accelerating pattern in
\(q_2\). This automaton can be decomposed into
automata accepting \(\Shuffle{w}\) for \(w\in \{ab,ba,bb,ca,cb\}\).
This illustrates the connection between accelerating patterns
and decomposability, which we now formalize.

\begin{theorep}\label{theo:union characterization}
	Let \(\A\) be a minimal  automaton that recognizes an ideal. We denote \(\lang{\A} = L\). If we denote \(m = \max\set{|w|}{w \in \Lmin}\) then the following are equivalent:
	\begin{enumerate}
		\item \(\A\) is prime.
		\item The size of \(\A\) is equal to \(m+1\).
		\item \(\A\) is linear and has no accelerating pattern.
	\end{enumerate}
\end{theorep}

\noindent
Before proving the above theorem we need to define the automata used in the decomposition. They are the ones that recognize languages of the form \(\Shuffle{\{w\}}\) where \(w\) is a word.

\subparagraph*{Principal Automata.} The \emph{principal automaton}
generated by \(w\in \Sigma^n\), written \(\A_w\), is the minimal
automaton that recognizes \(\Shuffle{\{w\}}\). Its set of states is
\(S_w = \set{q_i}{i \in \llbracket 0, n \rrbracket}\), its alphabet is
\(\Sigma\),  its initial state is \(q_0\), its final state is \(\{q_n\}\) and its set of transitions \(\delta_w\) is 
\[\set{(q_{i-1},w_{i},q_{i})}{i \in \setint{n-1}} \ \cup \bigcup\limits_{i \in
  \llbracket 0, n-1\rrbracket} \set{(q_i,a,q_i)}{a \in \Sigma
  \setminus\{w_{i + 1}\}} \ \cup
  \ \set{(q_n,a,q_n)}{a\in \Sigma}.\] This definition is illustrated in
Fig.\ref{fig: Auto A_w}.

\begin{figure}
	\centering
 	\includegraphics[width=0.5\textwidth]{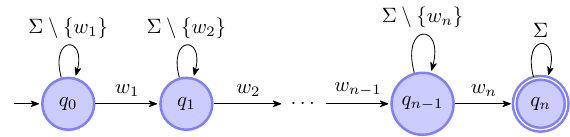}
	\caption{Example of \(\A_w\) when \(w = w_1\dots w_n\).}\label{fig: Auto A_w}
\end{figure}

\begin{proof}
  The equivalence between (1.) and (2.) has been proved in the paper.
  We assume that \(|\A| \geq 1\). Otherwise \(\A\) is prime, linear
  with no accelerating pattern.  Let \(\A = \AutoLin\) be a minimal
  linear automaton recognizing an ideal and satisfying \(q_i\prec
  q_{i+1}\).
  
  First, we prove that (3.) implies (2.). Assume that \(\A\) has  no
  accelerating pattern.  Then, we
  can set \(u= a_1 \dots a_{n}\) with \(a_i\in \Sigma_{i-1,i}\) such
  that for all \(j < i-1, a_i \notin \Sigma_{j,i}\).

  Using the claim of the equivalence proof between (1.) and (2.),
  \(m\leq n\). In order to prove that \(|\A|=m+1\), that is \(n=m\),
  it remains to prove that \(n \leq m\).  Since for all \(i \in
  \setint{n}\), one has \(a_i \in \Sigma_{i-1,i}\) and, by definition
  of \(\delta\), one has \(\delta(q_0,u)= q_n\). Therefore, since
  \(q_n\) is the final state, \(u \in \lang{\A}=L\). Now we will prove
  that no strict subword of \(u\) belongs to \(L\). Let \(u'\in
  \Sigma^t\) be a strict subword of \(u\). Since \(n\geq 1\), \(u\neq
  \varepsilon\). Consequently there exists a strictly increasing
  sequence \((\tau(i))_{i\in\setint{t}}\) such that for all \(i \in
  \setint{t}, u'_i = u_{\tau(i)}\). Let \(k = \max\set{i \in
    \setint{t}}{\tau(i) = i}\) if this set is non-empty, and \(k=0\)
  otherwise.

  \begin{itemize}
    \item If \(k=t\), then \(\delta(q_0,u')=q_k\). Since
      \(k<m\leq n\), \(q_k\neq q_n\). Therefore \(u'\notin
      \lang{\A}\).
    \item If \(k < t\) let \(q_i,q_j \in \setint{n}\) such that \(q_i
      \prec q_{j-1}\). Since \(\A\) is linear, one has
      \(\delta(q_i,a_j) \prec q_j\). Therefore, there exists \(v \in
      \Sigma^*\) such that \(\delta(q_i,v)=q_{j-1}\). Consequently,
      since in a linear automaton \(\preccurlyeq\) is a total order,
      Proposition~\ref{prop:new prop on ideals}.\ref{lemma: min
        automata preccurlyeq equality} provides that \(\delta(q_i,a_j)
      \preccurlyeq \delta(q_{j-1},a_j) = q_j\). By hypothesis \(a_j
      \notin \Sigma_{i,j}\), therefore \(\delta(q_i,a_j) \neq
      q_j\). Hence by definition of \(\prec\), one has that
      \(\delta(q_i,a_j) \prec q_j\). By definition of \(k\), one has
      \(\delta(q_0,w'_1\dots w'_k)= q_k\) and \( q_k \prec
      q_{\tau(k+1)-1} \). Now, using \(t - k\) times that
      \(\delta(q_i,a_j) \prec q_j\) when \(q_i\prec q_{j-1}\) , we
      obtain that \(\delta(q_0,u') \prec q_{(\tau(t))}\). Thus \(u'
      \notin \lang{\A}\).
  \end{itemize}
  We just proved that \(u\in L\) and no strict subword of \(u\)
  belongs to \(L\). Consequently \(u\in \Lmin\), proving that \(m=n\).

  Secondly, we prove that (2.) implies (3.). The proof is by
  contraposition. If \(\A\) is not linear, by the construction of \(\Lmin\) done 
  in~\cite{heam2002shuffle} one has that \(m < |\A|\) since \(m \leq \max\set{\norm{q}}{q \in S}\). 
  Now we consider \(\A\) to be linear. 
  Assume there exists \(i \in \setint{n}\) such that
  for all \(a \in \Sigma_{i-1,i}\) there exists \(j \in \setint{i-1}\)
  satisfying \(a \in \Sigma_{j,i}\). We will prove that \(m \leq n\).
  Let \(u \in \lang{\A} \cap \Sigma^{n}\).

  If there exists \(k< l \in \setint{n}\) such that
  \(\delta(q_0,u_1\dots u_k) = \delta(q_0,u_1\dots u_l)\), then one
  has \(\delta(q_0,u_1\dots u_k u_{l+1} u_{l+2} \dots u_{n}) =
  \delta(q_0,u)\).  Consequently \(u_1\dots u_k u_{l+1} u_{l+2} \dots
  u_{n} \in \lang{\A}\). Thus, by Proposition~\ref{prop: hold prop on
    ideals}.\ref{lemma:no subwords of Lmin in L}, \(u \notin \Lmin\).

  Otherwise, the pigeon hole principle and the structure of \(\A\)
  implies that for all \(k \in \setint{n}\), one has \(u_k \in
  \Sigma_{k-1,k}\). By hypothesis, there exists \(j \in \setint{i-1}\)
  such that \(u_i \in \Sigma_{j,i}\). Hence by the definition of
  \(\delta\), one has \(\delta(q_0,u_1\dots u_{j-1} u_i u_{i+1} \dots
  u_n) = q_n\). Therefore \(u_1\dots u_{j-1} u_i u_{i+1} \dots u_n\)
  is a subword of \(u\) that belongs to \(\lang{\A}\). Hence, using
  Proposition~\ref{lemma:no subwords of Lmin in L}, one has \(u
  \notin \Lmin\). Thus, by definition of \(m\), one has \(m
  \neq n\), which concludes the proof.
\end{proof}

\begin{proof}
	We show that \((1.)\) and \((2.)\) of theorem~\ref{theo:union characterization} are equivalent.
	The equivalence between \((2.)\) and \((3.)\) is proved in the appendix.

Let \(\wmax \in \Lmin\) such that \(|\wmax| = m\). We first claim that
every automaton of size smaller than or equal to \(m\) recognizing
\(\wmax\) also recognizes at least one word that is not in
\(L\).
Let \(\A' = \Auto\) be an automaton such that \(|\A'| \leq m\)
and \(\wmax \in \lang{\A}\). Let \((q_i)_{i\in
  \llbracket 0, m \rrbracket} \in S^{m+1}\) be the states visited by
the accepting path on \(w\).
Then \(q_0 = \iota\), \(q_m\in F\) and \(\delta(q_{i-1},\wmax_i) =
q_i\) for
all \(i \in \setint{m}\).
By a cardinality argument, there exist two \(i, j \in \llbracket 0, m
\rrbracket\) such that \(i < j\) and  \(q_i = q_j\). Let \(w'  =
\wmax_1 \dots \wmax_{i-1} \wmax_i \wmax_{j+1} \wmax_{j+2} \dots
\wmax_m\). By construction,  \(\delta(\iota,w') = q_m\), proving that
\(\wmax^\prime \in \lang{\A^\prime}\). But \(w^\prime\) is a subword of
  \(\wmax\). By Proposition~\ref{prop: hold prop on ideals}.\ref{lemma:no
    subwords of Lmin in L}, \(w' \notin L\), proving the
  claim. Hence \(|\A| \geq m + 1\) always hold.

  Let us now show that \((2.)\) implies \((1.)\). Assume now that \(|\A| = m+1\). Towards building a contradiction,
  assume that \(\A\) admits a
  decomposition \((\A_i)_{i\in \setint{k}}\). Since \(\wmax\in\lang{\A}\),
  there exists  \(i \in \setint{k}\) such that \(\wmax \in
  \lang{\A_i}\). Using the claim, either \(|\A_i|\geq m+1\)
  or \(\A_i\) recognizes a word that is not in \(\lang{\A}\),
  and both cases contradict the fact that \(\A_i\)
  is a component of a union-decomposition of \(\A\).
  Therefore \(\A\) is prime, proving that (2.)
  implies (1.).

  To finish we show that \((1.)\) implies \((2.)\). Assume that the size of \(\A\) is not \(m+1\). Since
  \(\wmax\in \lang{\A}\), using the claim, \(|\A|\geq m+1\). Therefore 
  \(|\A| > m+1\). But, one has
  \[\Shuffle{\lang{\A}} = \bigcup\limits_{u \in \Lmin} \Shuffle{\{u\}}\enspace.\]
 By construction of principal automata, \(|\A_u|=|u| +
 1\). Therefore\((\A_u)_{u \in \Lmin}\) is a decomposition of \(\A\)
 in \(|\Lmin|\) linear automata that recognize ideals, proving that
 \(\A\) is not prime, which proves that that (1.)
 implies (2.). \qedhere
\end{proof}

\noindent
Note that the proof of equivalence of \((1.)\) and \((2.)\) is
constructive: whenever \(\A\) is composite we can first compute \(\Lmin\) using the algorithm in~\cite{heam2002shuffle}, and then the present proof yields a decomposition into \(|\Lmin|\) prime linear automata that recognize ideals.

\section{Complexity results}\label{sec: Decomposition in prime automaton}

In this section, we show that the problems of deciding whether a
minimal automaton recognizing an ideal is decomposable with respect to 
intersection and union are both in NL. 
We also give bounds on the size of the decomposition into prime automata obtained by iterating the algorithms.

\begin{figure}
	\noindent
	\begin{minipage}[t]{0.48\textwidth}
		\begin{algorithm}[H]
			\textbf{InterComposite}\((\A)\)
			
			\eIf{\textbf{\rm\bf NonLinear}\((\A)\)}
			{
				return(\textbf{True})
			}
			{
				\textbf{guess} \(q,r \in S\)
				
				return(\textbf{Next}\((q,r) \ \wedge \ \)\textbf{Damping}\((q,r)\))
			}
		\end{algorithm}
		\begin{algorithm}[H]
			\textbf{NoDamping}\((q_i,q_j)\)
			
			\For{\(a \in \Sigma\)}
			{
			\If{ \(\delta(q_i,a) \in \{q_i,q_j\} \wedge \delta(q_j,a) \neq q_j\)} 
		{
			return(\textbf{True}) 
		}
	}
	return(\textbf{False})
\end{algorithm}
\end{minipage}
\hfill
\begin{minipage}[t]{0.48\textwidth}

\begin{algorithm}[H]
	\textbf{NonLinear}\((\A)\)
	
	\textbf{guess} \(q,r \in S\)
	
	return(\textbf{NoPath}\((q,r) \ \wedge \ \)\textbf{NoPath}(\(r,q\)))
\end{algorithm}
\begin{algorithm}[H]
	\textbf{NoNext}\((q_i,q_j)\)
	
	\eIf{{\textbf{\rm\bf NoPath}}\((q_i,q_j) \vee i = j\)}
	{
		return(\textbf{True})
	}
	{
		\textbf{guess} \(q \in S\setminus\{q_i,q_j\}\)
		
		\If{\textbf{\rm\bf Path}\((q_i,q)
			\ \wedge\) \textbf{\rm \bf Path}\((q,q_j)\)} 
		{
			return(\textbf{True}) 
		}
		return(\textbf{False})
	}
\end{algorithm}
\end{minipage}
\caption{An NL algorithm deciding 
whether a minimal automaton \(\A\) recognizing an ideal
is composite for intersection,
together with auxiliary NL procedures
performing local structural checks:
detecting a damping pattern between two states,
determining whether an automaton is linear,
and checking whether two states are consecutive.}
\label{fig:nlogspace algorithms}
\end{figure}

\maintheorem*

\begin{proof}
	Since the 
        Immerman-Szelepcsényi theorem~\cite{immerman1988nondeterministic}
        states that NL = co-NL,
        it suffices to design an NL algorithm deciding whether \(\A\) is
        composite.
        To this end, we use five procedures:
        {\bf InterComposite}, {\bf NoPath},  {\bf NonLinear}, {\bf NoNext} and
        {\bf NoDamping}, depicted in
        Fig.~\ref{fig:nlogspace algorithms}.
        
        Our main procedure, {\bf InterComposite},
        is based on the characterization of intersection-prime automata
        given by Theorems~\ref{theo:non-linear are decomposable} and~\ref{prop: decomposition of linear automata}:
        \(\A\) is composite if and only if it is non-linear
        or admits a damping pattern.
        Accordingly, the algorithm first invokes {\bf NonLinear},
        which determines whether \(\A\) is non-linear.
        If this is not the case,
        it nondeterministically guesses two states,
        invokes {\bf NoNext} to test whether they are consecutive, and then {\bf Damping} to check for the existence of a
        damping pattern.
        We now describe each subroutine in detail
        and show that they all operate in logarithmic space.
        \begin{itemize}
        	\item 
	        {\bf Path} and {\bf NoPath}
	        respectively decide whether there is a path or no path
	        between two given states.
	        It is known that both problems are in NL~\cite{immerman1988nondeterministic},
	        therefore we do not redefine the corresponding routines here.
	        \item 
	        {\bf NonLinear} takes as input an
	        automaton and accepts if it is non-linear.
	        It guesses two
	        states and returns {\bf True}
	        if and only if they are incomparable,
	        using {\bf NoPath} as a subroutine to perform
	        the reachability checks.
	        Since only two state indices need to be stored
	        in addition to the workspace
	        of the NL {\bf NoPath} procedure, this algorithm is in~NL.
	        \item 
	       {\bf NoNext} decides whether, in a linear
	        automaton,
	        \(q_j\) is not the direct successor of \(q_i\)
	        for the total order on states.
	        The first {\bf if} statement relies on {\bf NoPath}
	        to check whether
	        \(q_i \not\preccurlyeq q_j\) or \(q_i = q_j\),
	        both of which witness that \(q_j\)
	        is not the direct successor of \(q_i\).
	        Otherwise, since
	        \(q_i \prec q_j\), it remains to check whether there exists \(q\) such
	        that \(q_i\prec q \prec q_j\),
	        which is checked using the  {\bf Path} subroutine.
	        Since only three state indices need to be stored in addition to the workspace
	        of the NL
	        {\bf Path} and {\bf NoPath} procedures, this algorithm is in NL. And since Nl = co-Nl, there exists {\bf Next} an NL algorithm.
	        \item 
	        {\bf NoDamping} decides,
	        given \(i\) and \(j\),
	        whether there is no damping pattern between 
	        \(q_i\) and \(q_j\).
	        The {\bf if} statement checks whether
	        there exist a letter \(a\in \Sigma_{i,i} \cup \Sigma_{i,j}\)
	        which is not in
	        \(\Sigma_{j,j}\).
	        Since a damping pattern corresponds to the inclusion
	        \(\Sigma_{i,i} \cup \Sigma_{i,j} \subseteq \Sigma_{j,j}\),
	        the existence of such a letter witnesses its absence.
	        Only two state indices and one letter name need to be stored,
	        which can be done in logarithmic space. 
	        Moreover, since NL = co-NL, there also exists an NL algorithm for the complement
	        {\bf Damping} of {\bf NoDamping}. \qedhere
	    \end{itemize}
\end{proof}

\noindent
We show that these steps can be performed in linear time,
relying on a topological sorting (e.g., Kahn’s algorithm~\cite{DBLP:journals/cacm/Kahn62}).
Such a sorting can be used, first, to determine whether the automaton is linear.
If it is not, it can then be used to compute the separator set.
Otherwise, we simply traverse the states in topological order while searching for a damping pattern.

\constructionIntersection*

\begin{toappendix}
\constructionIntersection*
	\begin{proof}
Testing whether \(\A\) is linear can be done in time  \(\mathcal{O}(|\A||\Sigma|)\) using a topological sorting, for instance by Kahn's
algorithm \cite{DBLP:journals/cacm/Kahn62}. If \(\A\) is non-linear,
it is decomposable by  Theorem~\ref{theo:non-linear are
  decomposable}. Otherwise, the topological sort provides the rank of
each state and one can verify in linear time the presence of a damping
pattern. Proposition~\ref{prop: decomposition of linear automata}
allows then to conclude on the intersection primality of \(\A\).

If \(\A\) is composite, and non-linear, then finding the smallest
rank \(k\) verifying that \(|\set{q}{\norm{q} = k}| \geq 2\) can be done in
\(\mathcal{O}(|\A|)\) time steps using the topological sort.

  Let \(m\) be the number of states of rank \(k\). Then one can build
  arrays representing \(\mifa{q}\) for \(q\) a state of rank \(k\) in
  \(\mathcal{O}(|\A||\Sigma|m)\) time steps. Finally we create \(m\)
  copies of \(\A\) and modify them to be the family automata \(\A(\rho)\) in
  \(\mathcal{O}(|\A||\Sigma|m)\) steps. To do so one can modify the
  transitions of state by order of decreasing rank using an array to
  remember all the \(s^\rho_{q,a}\) of the construction. Therefore the time complexity in this case is
  \(\mathcal{O}(|\A||\Sigma|^2)\) since \(m\) is at most \(|\Sigma|\).

Otherwise, if \(\A\) linear and composite, one can compute the
automata \(\A_0\) and \(\A_1\) defined in
Section~\ref{subsec:DecomposableLinearAutomata}  in time \(\mathcal{O}(|\A||\Sigma|)\).
	\end{proof}
\end{toappendix}

\begin{figure}[t]
	\centering
	\includegraphics[width=0.6\textwidth]{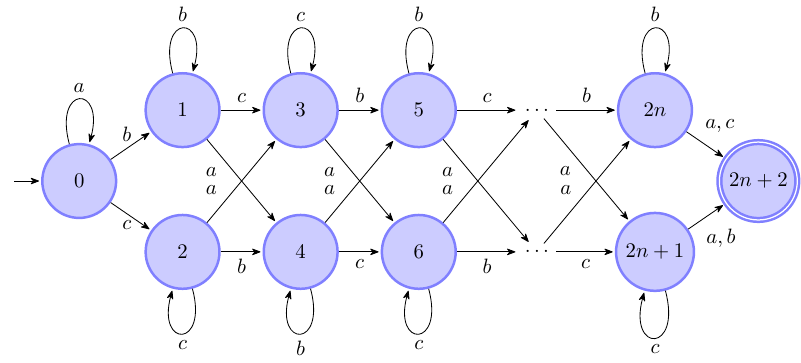}
	\caption{Example of the construction of \(\A_n\) non-linear that is decomposed by the algorithm in \(2^{|\A|/2-1}\) linear automata.}\label{fig:non_linear_to_exp_linear}
\end{figure}

\noindent
We now study the size of the decomposition obtained
by iteratively applying our constructions until reaching intersection-prime components.
We show that repeated application of our procedure on a non-linear automaton
yields at most exponentially many linear automata, and that this bound is tight.
Similarly, applying our procedure to a linear automaton
produces at most exponentially many prime components, and this bound is also tight.
In both cases, we construct families of automata witnessing the tightness
of the exponential bound.

\begin{proprep}\label{lemma: bound sup Decomposition in linear automata}
Let \(\A = \minAuto\) be an non-linear minimal automaton that
recognizes an ideal. Iterating the decomposition described in Theorem~\ref{theo:non-linear are decomposable} yields
at most \(2^{|\succsA{\A}{\sep}|}\) linear automata. 
\end{proprep}

\begin{proof}
	This lemma is shown by induction on \(|\A|\).

	The construction describes described in Theorem~\ref{theo:non-linear are decomposable} decompose \(\A\) in \((\A(\rho)_{\rho\in \setsep})\). First, let's show that for any \(\rho \in \setsep\), the decomposition in linear automata given by our algorithm is at most \(2^{|\succsA{\A}{\sep}| - |\setsep|}\) then we will conclude using a inequality on integers. Let \(\rho \in \setsep\). Lemma~\ref{lemma: preccurlyeq in A' and A} and Lemma~\ref{lemma:unicity above sep} says that if \(\A(\rho)\) is not linear, then its separator state is smaller than \(\rho\). Hence, since there are only trivial cycles in \(\A\) and since for all \(\rho' \in \setsep\setminus \{\rho\}\), one has \(\rho' \notin \succsA{\A(\rho)}{\rho}\) and \(\sep\notin \succsA{\A(\rho)}{\rho}\), the induction hypothesis says that \(\A(\rho)\) is decomposed in at most \(2^{|\succsA{\A}{\sep}| - |\setsep|}\) linear automata.
	Therefore, \(\A\) is decomposed in \(|\setsep|\) automata which, by induction, are decomposed in at most \(2^{|\succsA{\A}{\sep}| - |\setsep|}\) linear automata. Hence, one can conclude that \(\A\) is decomposed in at most \(|\setsep|2^{|\succsA{\A}{\sep}| - |\setsep|}\) linear automata. And since for all \(k \in \N\) one has \(k \leq 2^k\), iterating the construction in Theorem~\ref{theo:non-linear are decomposable} \(\A\) is decomposed in at most \(2^{|\succsA{\A}{\sep}|}\) linear automata.
\end{proof}

We conjecture that this bound can be lowered to
\(2^\frac{|\succsA{\A}{\sep}|}{2}\), and we know it cannot be further improved.
Indeed, for all \(n \in \N\), let \(\A_n\) be the automaton
in Fig.~\ref{fig:non_linear_to_exp_linear}.
This automaton has \(2n+2\) states, and is decomposed in \(2^{n}\) linear automata by iterating the
construction of Theorem~\ref{theo:non-linear are decomposable}, one for each path from state \(0\) to state \(2n+2\).

\begin{proprep}\label{lemma: Decomposition of linear in prime automata}
	Let \(\A\) be a linear automaton that recognizes an
    ideal. Iterating the decomposition described in 
    Theorem~\ref{prop: decomposition of linear automata} yields
    at most \(2^{|\A|}\) automata. 
\end{proprep}

\begin{proof}
	The proof is done by induction on the size of \(\A\).

	If \(\A\) is composed of a unique state, then \(\A\) is prime.

	Otherwise if \(2 \leq |\A|\) then, by Proposition~\ref{prop:
          decomposition of linear automata}, \(\A\) is decomposed in
        two linear automata of size \(|\A|-1\). Then, by induction
        hypothesis on these two automata, one has that \(|\A|\) is decomposed in at most \(2\times 2^{|A|-1} = 2^{|\A|}\) automata.  
\end{proof}

The bound of Proposition~\ref{lemma: Decomposition of linear
  in prime automata} is tight.
Let \(\A = \minAuto\) and \(\A' =
(S',\Sigma,\iota',\{q_f^\prime\},\delta')\) be two minimal automata
that recognize ideals. We define \(\A \cdot \A'=(S
\sqcup S'\setminus F, \Sigma, \iota, F', \delta'')\), where \(
\sqcup\) is the disjoint union (we rename states if necessary) and where
\[\delta'' = \set{(q,a,q')\in \delta}{q' \neq q_f} \cup \set{(q,a,\iota')}{(q,a,q_f) \in \delta} \cup \delta'\enspace.\]

We write, for an automaton \(\A\) that recognizes an ideal and \(n \in
\N\setminus\{0\}\), \(\A^1 = \A\) and inductively \(\A^n = \A^{n-1} \cdot \A\).
 Let \(\A, \A_0\) and \(\A_1\) be the automata of Fig.~\ref{fig:
   decomposition}. The size of \(\A^n\) is
 then \(3n+1\). Using iteratively Proposition~\ref{prop: decomposition of linear automata}, \(\A^n\) is decomposed in \(\{\A_{i(1)}\dots \A_{i(n)} \mid (i(j))_{j\in \setint{n}} \in \{0,1\}^n\}\), which contains \(2^n= \Omega(2^{|\A_i|})\) prime automata. 
Using the two previous propositions, one can conclude the following:

\begin{corollary}
	Let \(\A\) be a minimal automaton recognizing an ideal.
	Using algorithms of the previous sections,
	one can compute an intersection-decomposition of \(\A\)
	into at most \(2^{2|\A|}\) intersection-prime automata
	that recognize ideals.
\end{corollary}

\begin{figure}[t]
	\begin{algorithm}[H]
		\textbf{UnionComposite}\((\A)\)
		
		\eIf{\textbf{\rm\bf NonLinear}\((\A)\)}
		{
			return(\textbf{True})
		}
		{
			\textbf{guess} \(q,s \in S\) \tcp*{We guess a state that has the accelerating pattern}
			
			\eIf{\textbf{Next}\((q,s)\) \tcp*{Verify that that \(q = q_{i-1}\) and \(s = q_i\)}}
			{
				\For{\(a \in \Sigma\)}
				{
					\If{\(\delta(q,a) = s\) \tcp*{For all \(a \in \Sigma_{i-1,i}\)}} 
					{
						\textbf{guess} \(r \in S\) \tcp*{we try to guess a \(q_j\) such that \(a \in \Sigma_{j,i}\)}
						
						\If{\(\neg (
							\textbf{\rm\bf
								Path}(r,q) \wedge r \neq q \wedge \delta(r,a) = s)\)}
						{
							return(\textbf{False})
						}
					}
				}
				return(\textbf{True}) \tcp*{We accept if for all \(a \in \Sigma_{i-1,i}\) we found a \(q_j\)}
			}
			{
				return(\textbf{False})
			}
		}
	\end{algorithm}
	\caption{An NL algorithm deciding 
		whether a minimal automaton \(\A\) recognizing an ideal
		is composite for union.
		The procedure guesses candidate states witnessing the accelerating pattern
		and verifies locally the required transition constraints.}
	\label{fig:nlogspace union algorithm}
\end{figure}

We now investigate the complexity of deciding the primality for the
union decomposition. 

\constructionUnion*

\begin{proof}
	Since the 
	Immerman-Szelepcsényi theorem~\cite{immerman1988nondeterministic}
	states that NL = co-NL,
	it suffices to design an NL algorithm deciding whether \(\A\) is
	composite.
	To this end, we use five procedures:
	{\bf UnionComposite}, {\bf NoPath},  {\bf NonLinear} and {\bf Next},
	depicted in
	Fig.~\ref{fig:nlogspace algorithms} and Fig.~\ref{fig:nlogspace union algorithm}.
	We refer the reader to the proof of Theorem~\ref{theo:space complexity}
	for the detailed description of the subroutines and the fact that they are in NL.

  The main procedure {\bf UnionComposite} relies of the characterization
  of union-composite automata stated in Theorem~\ref{theo:union characterization}:
  \(\A\) is composite if and only if it is non-linear or has an accelerating pattern.
  Accordingly,  {\bf UnionComposite} first calls {\bf NonLinear} to check
  whether \(\A\) is non-linear.
  If it is the case, it accepts.
  Otherwise, it tries to guess a state \(s=q_i\)
  that may witness an accelerating
  pattern, together with its predecessor \(q=q_{i-1}\). 
  Next, it
  sequentially enumerates all the letters \(a\) in
  \(\Sigma\), and for each \(a\in\Sigma_{i-1,i}\)
  it nondeterministically guesses
  a state \(q_j \prec q_{i-1}\) such that \(a\in \Sigma_{j,i}\).
  If this succeeds for all letters,
  then the algorithm has found an accelerating pattern and accepts.
  Otherwise, the algorithm returns {\bf False}.
  At each step, a fixed number of state indexes and letter names must
  be stored, which can be achieved using a logarithmic space.
 \end{proof}

\section{Conclusion}
\label{sec:conclusion}

We studied decomposition problems
for automata recognizing ideal languages,
a subclass of aperiodic languages
corresponding to level \(1/2\)
of the Straubing-Thérien hierarchy.
We showed that both intersection and union primality
can be decided in NL, which contrasts with the general setting,
where no algorithms below EXPTIME are known.
We now discuss
how our results relate to wider open questions.

\subparagraph*{Are there efficient decomposition algorithms for aperiodic languages?}

\noindent
A natural direction for future work is to extend our techniques beyond ideals,
and to climb further in the Straubing--Thérien hierarchy.
Higher levels are obtained by iterating Boolean and polynomial closure operations.
The techniques developed in this work for both intersection and union decompositions
may provide a starting point for studying Boolean closures of ideals.

\subparagraph*{Which classes of regular languages are preserved by decomposition?}

\noindent
One notable aspect of our results is that decompositions
preserve the class of ideal languages:
whenever an automaton recognizing an ideal is decomposable (either as an intersection or as a union),
it can be decomposed into automata that also recognize ideals.
This mirrors what is known for permutation automata~\cite{jecker2025decomposing}.
In contrast, for automata recognizing finite languages,
existing decomposition techniques do not preserve the class~\cite{spenner2023decomposing}.
This contrast raises several natural questions:
can one prove that some finite languages necessarily require infinite languages in their decompositions?
More generally, can we characterize which subclasses of languages are closed under decomposition,
and which are not?

\subparagraph*{How to decompose infinite word languages?}

\noindent
In the context of model checking, it would be particularly worthwhile
to explore how the various decomposition results extend to infinite
words. Investigating the languages defined by LTL formulas, or
equivalently by \(FO(<)\), presents a natural and relevant challenge,
especially with the aim of simplifying specifications. A promising
starting point would be the first existential fragment of \(FO(<)\),
which, over finite words, corresponds to the ideals studied in this
paper.

\bibliography{biblio}

@book{DBLP:books/aw/HopcroftU79,
  author       = {John E. Hopcroft and
                  Jeffrey D. Ullman},
  title        = {Introduction to Automata Theory, Languages and Computation},
  publisher    = {Addison-Wesley},
  year         = {1979},
  url          = {https://doi.org/10.1145/568438.568455},
  doi          = {10.1145/568438.568455},
}

@article{DBLP:journals/tc/KamedaW70,
  author       = {Tsunehiko Kameda and
                  Peter Weiner},
  title        = {On the State Minimization of Nondeterministic Finite Automata},
  journal      = {{IEEE} Trans. Computers},
  volume       = {19},
  number       = {7},
  pages        = {617--627},
  year         = {1970},
  url          = {https://doi.org/10.1109/T-C.1970.222994},
  doi          = {10.1109/T-C.1970.222994},
}

@article{DBLP:journals/siamcomp/JiangR93,
  author       = {Tao Jiang and
                  Bala Ravikumar},
  title        = {Minimal {NFA} Problems are Hard},
  journal      = {{SIAM} J. Comput.},
  volume       = {22},
  number       = {6},
  pages        = {1117--1141},
  year         = {1993},
  url          = {https://doi.org/10.1137/0222067},
  doi          = {10.1137/0222067},
}

@article{brainerd1968minimalization,
  title        = {The minimalization of tree automata},
  author       = {Brainerd, Walter S},
  journal      = {Information and Control},
  volume       = {13},
  number       = {5},
  pages        = {484--491},
  year         = {1968},
  publisher    = {Elsevier},
  url          = {https://doi.org/10.1016/S0019-9958(68)90917-0},
  doi          = {10.1016/S0019-9958(68)90917-0},
}

@proceedings{DBLP:conf/compos/1997,
  editor       = {Willem P. de Roever and
                  Hans Langmaack and
                  Amir Pnueli},
  title        = {Compositionality: The Significant Difference, International Symposium,
                  COMPOS'97, Bad Malente, Germany, September 8-12, 1997. Revised Lectures},
  series       = {Lecture Notes in Computer Science},
  volume       = {1536},
  publisher    = {Springer},
  year         = {1998},
  url          = {https://doi.org/10.1007/3-540-49213-5},
  doi          = {10.1007/3-540-49213-5},
}

@article{kupferman2015prime,
  title        = {Prime languages},
  author       = {Kupferman, Orna and Mosheiff, Jonathan},
  journal      = {Information and Computation},
  volume       = {240},
  pages        = {90--107},
  year         = {2015},
  publisher    = {Elsevier},
  url          = {https://doi.org/10.1016/j.ic.2014.09.010},
  doi          = {10.1016/J.IC.2014.09.010},
}

@inproceedings{DBLP:conf/cav/SternD97,
  author       = {Ulrich Stern and
                  David L. Dill},
  title        = {Parallelizing the Mur\emph{phi} Verifier},
  booktitle    = {{CAV} '97,
                  Haifa, Israel, June 22-25, 1997, Proceedings},
  series       = {Lecture Notes in Computer Science},
  pages        = {256--278},
  publisher    = {Springer},
  year         = {1997},
  url          = {https://doi.org/10.1007/3-540-63166-6\_26},
  doi          = {10.1007/3-540-63166-6\_26},
}

@incollection{DBLP:books/sp/18/BarnatBDLPPR18,
  author       = {Jiri Barnat and
                  Vincent Bloemen and
                  Alexandre Duret{-}Lutz and
                  Alfons Laarman and
                  Laure Petrucci and
                  Jaco van de Pol and
                  Etienne Renault},
  title        = {Parallel Model Checking Algorithms for Linear-Time Temporal Logic},
  booktitle    = {Handbook of Parallel Constraint Reasoning},
  pages        = {457--507},
  publisher    = {Springer},
  year         = {2018},
  url          = {https://doi.org/10.1007/978-3-319-63516-3\_12},
  doi          = {10.1007/978-3-319-63516-3\_12},
}

@inproceedings{DBLP:conf/formats/DalsgaardLLOP12,
  author       = {Andreas Engelbredt Dalsgaard and
                  Alfons Laarman and
                  Kim G. Larsen and
                  Mads Chr. Olesen and
                  Jaco van de Pol},
  title        = {Multi-core Reachability for Timed Automata},
  booktitle    = {{FORMATS} 2012, London, UK, September 18-20, 2012. Proceedings},
  series       = {Lecture Notes in Computer Science},
  pages        = {91--106},
  publisher    = {Springer},
  year         = {2012},
  url          = {https://doi.org/10.1007/978-3-642-33365-1\_8},
  doi          = {10.1007/978-3-642-33365-1\_8},
}

@inproceedings{spenner2023decomposing,
  author       = {Daniel Alexander Spenner},
  title        = {Decomposing Finite Languages},
  booktitle    = {{MFCS} 2023, Bordeaux, France, August 28 - September 1, 2023},
  series       = {LIPIcs},
  volume       = {272},
  pages        = {83:1--83:14},
  publisher    = {Schloss Dagstuhl - Leibniz-Zentrum f{\"{u}}r Informatik},
  year         = {2023},
  url          = {https://doi.org/10.4230/LIPIcs.MFCS.2023.83},
  doi          = {10.4230/LIPICS.MFCS.2023.83},
}

@inproceedings{DBLP:conf/mfcs/JeckerKM20,
  author       = {Isma{\"{e}}l Jecker and
                  Orna Kupferman and
                  Nicolas Mazzocchi},
  title        = {Unary Prime Languages},
  booktitle    = {{MFCS 2020}, Prague, Czech Republic, August 24-28, 2020},
  series       = {LIPIcs},
  pages        = {51:1--51:12},
  publisher    = {Schloss Dagstuhl - Leibniz-Zentrum f{\"{u}}r Informatik},
  year         = {2020},
  url          = {https://doi.org/10.4230/LIPIcs.MFCS.2020.51},
  doi          = {10.4230/LIPICS.MFCS.2020.51},
}

@article{straubing1985finite,
  title        = {Finite semigroup varieties of the form {V}*{D}},
  author       = {Straubing, Howard},
  journal      = {Journal of Pure and Applied Algebra},
  volume       = {36},
  pages        = {53--94},
  year         = {1985},
  publisher    = {Elsevier},
  url          = {https://doi.org/10.1016/S0304-3975(96)00297-6},
  doi          = {10.1016/S0304-3975(96)00297-6},
}

@article{therien1981classification,
  title        = {Classification of finite monoids: the language approach},
  author       = {Th{\'e}rien, Denis},
  journal      = {Theoretical Computer Science},
  volume       = {14},
  number       = {2},
  pages        = {195--208},
  year         = {1981},
  publisher    = {Elsevier},
  url          = {https://doi.org/10.1016/0304-3975(81)90057-8},
  doi          = {10.1016/0304-3975(81)90057-8},
}

@inproceedings{barloy2022regular,
  title        = {The Regular Languages of First-Order Logic with One Alternation},
  author       = {Barloy, Corentin and Cadilhac, Micha{\"e}l and Paperman, Charles and Zeume, Thomas},
  booktitle    = {{LICS} '22: 37th Annual {ACM/IEEE} Symposium on Logic in Computer
                  Science, Haifa, Israel, August 2 - 5, 2022},
  year         = {2022}, 
  url          = {https://doi.org/10.1145/3531130.3533371},
  doi          = {10.1145/3531130.3533371},
}

@article{heam2002shuffle,
  title        = {On shuffle ideals},
  author       = {H{\'e}am, Pierre-Cyrille},
  journal      = {RAIRO-Theoretical Informatics and Applications},
  volume       = {36},
  number       = {4},
  pages        = {359--384},
  year         = {2002},
  publisher    = {EDP Sciences},
  url          = {https://doi.org/10.1051/ita:2003002},
  doi          = {10.1051/ITA:2003002},
}

@article{straubing1988partially,
  title        = {Partially ordered finite monoids and a theorem of {I}. {S}imon},
  author       = {Straubing, Howard and Th{\'e}rien, Denis},
  journal      = {Journal of Algebra},
  volume       = {119},
  number       = {2},
  pages        = {393--399},
  year         = {1988},
  publisher    = {Elsevier},
  url          = {PAS TROUVE},
  doi          = {10.1016/0021-8693(88)90067-1},
}

@article{immerman1988nondeterministic,
  title        = {Nondeterministic space is closed under complementation},
  author       = {Immerman, Neil},
  journal      = {{SIAM} Journal on computing},
  volume       = {17},
  number       = {5},
  pages        = {935--938},
  year         = {1988},
  publisher    = {SIAM},
  url          = {https://doi.org/10.1109/SCT.1988.5270},
  doi          = {10.1109/SCT.1988.5270},
}

@article{abdulla2004using,
  title        = {Using forward reachability analysis for verification of lossy channel systems},
  author       = {Abdulla, Parosh Aziz and Collomb-Annichini, Aurore and Bouajjani, Ahmed and Jonsson, Bengt},
  journal      = {Formal Methods in System Design},
  volume       = {25},
  number       = {1},
  pages        = {39--65},
  year         = {2004},
  publisher    = {Springer},
  url          = {https://doi.org/10.1023/B:FORM.0000033962.51898.1a},
  doi          = {10.1023/B:FORM.0000033962.51898.1A},
}

@article{jecker2025decomposing,
  author       = {Isma{\"{e}}l Jecker and
                  Nicolas Mazzocchi and
                  Petra Wolf},
  editor       = {Serge Haddad and
                  Daniele Varacca},
  title        = {Decomposing Permutation Automata},
  journal      = {{CONCUR} 2021},
  publisher    = {Schloss Dagstuhl - Leibniz-Zentrum f{\"{u}}r Informatik},
  year         = {2021},
  url          = {https://doi.org/10.4230/LIPIcs.CONCUR.2021.18},
  doi          = {10.4230/LIPICS.CONCUR.2021.18},
}

@inproceedings{karandikar2014state,
  title        = {On the state complexity of closures and interiors of regular languages with subwords},
  author       = {Karandikar, Prateek and Schnoebelen, Philippe},
  booktitle    = {International Workshop on Descriptional Complexity of Formal Systems},
  pages        = {234--245},
  year         = {2014},
  organization = {Springer},
  url          = {https://doi.org/10.1007/978-3-319-09704-6\_21},
  doi          = {10.1007/978-3-319-09704-6\_21}
}

@article{karandikar2015index,
  title        = {On the index of {S}imon's congruence for piecewise testability},
  author       = {Karandikar, Prateek and Kufleitner, Manfred and Philipe Schnoebelen},
  journal      = {Information Processing Letters},
  volume       = {115},
  number       = {4},
  pages        = {515--519},
  year         = {2015},
  publisher    = {Elsevier},
  url          = {https://doi.org/10.1016/j.ipl.2014.11.008},
  doi          = {10.1016/J.IPL.2014.11.008},
}

@article{okhotin2010state,
  title        = {On the state complexity of scattered substrings and superstrings},
  author       = {Okhotin, Alexander},
  journal      = {Fundamenta Informaticae},
  volume       = {99},
  number       = {3},
  pages        = {325--338},
  year         = {2010},
  publisher    = {SAGE Publications Sage UK: London, England},
  url          = {https://doi.org/10.3233/FI-2010-252},
  doi          = {10.3233/FI-2010-252},
}

@inproceedings{DBLP:conf/csl/KarandikarS16,
  author       = {Prateek Karandikar and
                  Philippe Schnoebelen},
  title        = {The Height of Piecewise-Testable Languages with Applications in Logical
                  Complexity},
  booktitle    = {{CSL} 2016,
                  Marseille, France, August 29 - September 1, 2016},
  series       = {LIPIcs},
  pages        = {37:1--37:22},
  publisher    = {Schloss Dagstuhl - Leibniz-Zentrum f{\"{u}}r Informatik},
  year         = {2016},
  url          = {https://doi.org/10.4230/LIPIcs.CSL.2016.37},
  doi          = {10.4230/LIPICS.CSL.2016.37},
}

@article{DBLP:journals/lmcs/MasopustK21,
  author       = {Tom{\'{a}}s Masopust and
                  Markus Kr{\"{o}}tzsch},
  title        = {Partially Ordered Automata and Piecewise Testability},
  journal      = {Log. Methods Comput. Sci.},
  volume       = {17},
  number       = {2},
  year         = {2021},
  url          = {https://lmcs.episciences.org/7475},
  doi          = {10.23638/LMCS-17(2:14)2021},
}

@article{DBLP:journals/dm/Klima11,
  author       = {Ondrej Kl{\'{\i}}ma},
  title        = {Piecewise testable languages via combinatorics on words},
  journal      = {Discret. Math.},
  volume       = {311},
  number       = {20},
  pages        = {2124--2127},
  year         = {2011},
  url          = {https://doi.org/10.1016/j.disc.2011.06.013},
  doi          = {10.1016/J.DISC.2011.06.013},
}

@Article{Haines,
  author       = {Leonard H. Haines},
  title        = {On free monoids partially ordered by embedding.},
  journal      = {Journal of Combinatorial Theory},
  year         = {1969},
  volume       = {6},
  pages        = {94-98},
  doi          = {10.1016/S0021-9800(69)80111-0},
  url          = {https://doi.org/10.1016/S0021-9800(69)80111-0},
}

@inproceedings{DBLP:conf/automata/Simon75,
  author       = {Imre Simon},
  title        = {Piecewise testable events},
  booktitle    = {Automata Theory and Formal Languages, 2nd {GI} Conference, Kaiserslautern,
                  May 20-23, 1975},
  series       = {Lecture Notes in Computer Science},
  pages        = {214--222},
  publisher    = {Springer},
  year         = {1975},
  url          = {https://doi.org/10.1007/3-540-07407-4\_23},
  doi          = {10.1007/3-540-07407-4\_23},
}

@Article{higman,
  author       = {Graham Higman},
  title        = {Ordering by divisibility in abstract algebras},
  journal      = {Proc. Lond. Math. Soc.},
  year         = {1952},
  OPTkey       = {},
  volume       = {3},
  number       = {2},
  pages        = {326-336},
  OPTmonth     = {},
  OPTnote      = {},
  OPTannote    = {},
  url          = {https://doi.org/10.1112/PLMS%2FS3-2.1.326},
  doi          = {10.1112/PLMS/S3-2.1.326},
}

@inproceedings{DBLP:conf/icalp/PinW95,
  author       = {Jean{-}Eric Pin and
                  Pascal Weil},
  title        = {Polynomial Closure and Unambiguous Product},
  booktitle    = {{ICALP}95, Szeged, Hungary, July 10-14, 1995, Proceedings},
  series       = {Lecture Notes in Computer Science},
  pages        = {348--359},
  publisher    = {Springer},
  year         = {1995},
  url          = {https://doi.org/10.1007/3-540-60084-1\_87},
  doi          = {10.1007/3-540-60084-1\_87},
}

@article{DBLP:journals/cacm/Kahn62,
  author       = {Arthur B. Kahn},
  title        = {Topological sorting of large networks},
  journal      = {Commun. {ACM}},
  volume       = {5},
  number       = {11},
  pages        = {558--562},
  year         = {1962},
  url          = {https://doi.org/10.1145/368996.369025},
  doi          = {10.1145/368996.369025},
}

@article{DBLP:journals/corr/abs-2605-07031,
  author       = {Daniel Alexander Spenner},
  title        = {Deciding DFA-Primality is NP-Hard},
  journal      = {CoRR},
  volume       = {abs/2605.07031},
  year         = {2026},
  url          = {https://doi.org/10.48550/arXiv.2605.07031},
  doi          = {10.48550/ARXIV.2605.07031},
  eprinttype   = {arXiv},
  eprint       = {2605.07031},
  bibsource    = {dblp computer science bibliography, https://dblp.org},
}
\end{document}